\documentclass[runningheads]{llncs}
\usepackage[T1]{fontenc}
\def\doi#1{\href{https://doi.org/\detokenize{#1}}{\url{https://doi.org/\detokenize{#1}}}}

\usepackage{listings}
\usepackage{graphicx}
\usepackage{amsmath}
\usepackage{amssymb}
\usepackage{hyperref}
\usepackage[capitalise]{cleveref}
\usepackage[dvipsnames]{xcolor}
\usepackage[shortlabels]{enumitem}
\usepackage[ruled,vlined]{algorithm2e}  
\usepackage{multicol}
\usepackage{multirow}
\usepackage{comment}
\usepackage{thm-restate}
 
\usepackage{fullpage}
\usepackage{caption}
\usepackage{subcaption}

\begin{document}

\title{Safety and Completeness in Flow Decompositions for RNA Assembly~\thanks{We thank Romeo Rizzi and Edin Husi\'c for helpful discussions. This work was partially funded by the European Research Council (ERC) under the European Union's Horizon 2020 research and innovation programme (grant agreement No.~851093, SAFEBIO) and partially by the Academy of Finland (grants No.~322595, 328877).}
}

\titlerunning{Safe and Comp. in Flow Decomp. for RNA Assembly} 
\author{Shahbaz Khan\inst{1}\orcidID{0000-0001-9352-0088} \and
Milla Kortelainen \inst{2} \orcidID{0000-0003-1590-0987} \and
Manuel C\'aceres \inst{2}\orcidID{0000-0003-0235-6951} \and
Lucia Williams \inst{3}\orcidID{0000-0003-3785-0247} \and
Alexandru I. Tomescu\inst{2}\orcidID{0000-0002-5747-8350}}

\authorrunning{S. Khan et al.} 

\institute{Department of Computer Science and Engineering, Indian Institute of Technology Roorkee, India \\
\email{shahbaz.khan@cs.iitr.ac.in\footnote{A major part of the work was done while affiliated with the Department of computer Science, University of Helsinki.}}
\and
Department of Computer Science, University of Helsinki, Finland \\
\email{\{shahbaz.khan,milla.kortelainen,manuel.caceresreyes,alexandru.tomescu\}@helsinki.fi}
\and
School of Computing, Montana State University, USA \\
\email{luciawilliams@montana.edu}
}

\bibliographystyle{splncs04}

\newcommand{\sbzk}[1]{\textcolor{cyan}{(sbzk: #1)}}
\newcommand{\Sbzk}[1]{\textcolor{cyan}{#1}}
\newcommand{\alex}[1]{\textcolor{brown}{#1}}
\newcommand{\alexcom}[1]{\textcolor{brown}{(Alex: #1)}}
\newcommand{\lucycom}[1]{{\color{Green} (Lucy: #1)}}
\newcommand{\arielcom}[1]{\textcolor{teal}{(Ariel: #1)}}
\newcommand{\millacom}[1]{\textcolor{olive}{(Milla: #1)}}

\maketitle

\begin{abstract}

Decomposing a network flow into weighted paths is a problem with numerous
applications, ranging from networking, transportation planning to bioinformatics.
In some applications we look for any decomposition that is optimal with respect
to some property, such as number of paths used, robustness to edge deletion,
or length of the longest path.
However, in many bioinformatic applications, we seek a specific
decomposition where the paths correspond to some underlying data that generated the flow.
For realistic inputs,
no optimization criteria can be guaranteed to uniquely
identify the correct decomposition. Therefore, we propose to instead report the {\em safe} paths, which are subpaths of at least one path in every flow decomposition.

~~~~Recently, Ma, Zheng, and Kingsford [WABI 2020] addressed the existence
of multiple optimal solutions in a probabilistic framework, which is referred to as
\emph{non-identifiability}.
In a follow-up work [RECOMB 2021], they gave a quadratic-time algorithm based on a {\em global} criterion for solving a problem called AND-Quant, which 
generalizes the problem of reporting whether a given path is safe.

~~~~~~In this work, we give the first {\em local} characterization of safe paths for flow decompositions in directed acyclic graphs (DAGs), leading to a practical algorithm for finding the \emph{complete} set of safe paths. We additionally evaluated our algorithms against the trivial safe algorithms (unitigs, extended unitigs) and the popularly used heuristic (greedy-width) for flow decomposition on RNA transcripts datasets. We find that despite maintaining perfect precision the  safe and complete algorithm reports significantly higher coverage ($\approx$50\% more) as compared to trivial safe algorithms. The greedy-width algorithm though reporting a better coverage, reports significantly lower precision on complex graphs (for genes expressing a large number of transcripts). Overall, our safe and complete algorithm outperforms (by $\approx$20\%) greedy-width on a unified metric (F-Score) considering both coverage and precision when the evaluated dataset has significant number of complex graphs. Moreover, it also has superior time (3$-$5$\times$) and space efficiency (1.2$-$2.2$\times$), resulting in a better and more practical approach for bioinformatics applications of flow decomposition.


\end{abstract}

\newpage
\section{Introduction}

Network flows are a central topic in computer science, enabling us to define problems
with countless practical applications. Assuming that the flow network has a unique source $s$ and a
unique sink $t$, every flow can be decomposed into a collection of weighted
$s$-$t$ paths and cycles~\cite{FordF10}; for directed acyclic graphs (DAGs), such a decomposition contains
only paths. One application of such a path (and cycle) view of a flow is to indicate how to
optimally route information or goods
from $s$ to $t$. For example, flow decomposition is a key step in
network routing problems~\cite{hong2013achieving,cohen2014effect,hartman2012split,mumey2015parity} and
transportation problems~\cite{Ohst:2015aa,Olsen:2020aa}.
Finding the decomposition with the minimum number of paths and {\em possibly} cycles (or {\em minimum flow decomposition}) is NP-hard, even if the flow network is a DAG~\cite{vatinlen2008simple}. On the theoretical side, this hardness result led to research on approximation algorithms~\cite{hartman2012split,SUPPAKITPAISARN2016367,pienkosz2015integral,mumey2015parity,baier2002k,baier2005k}, and FPT algorithms~\cite{kloster2018practical}. On the practical side, many approaches usually employ a standard \emph{greedy-width} heuristic~\cite{vatinlen2008simple}, of repeatedly removing an $s$-$t$ path carrying the most amount of flow. Another pseudo-polynomial-time heuristic called \emph{Catfish} was recently proposed by~\cite{shao2017theory}, which tries to iteratively simplify the graph so that smaller decompositions can be found.

On the other hand,
we may observe a flow network built by superimposing a set of weighted paths,
and seek the decomposition corresponding to that underlying set of paths and weights.
This is the decomposition sought by the more
recent and prominent application of reconstructing biological
sequences~(\emph{RNA transcripts}~\cite{pertea2015stringtie,tomescu2013novel,gatter2019ryuto,bernard2013flipflop,TomescuGPRKM15,williams2019rna}
or \emph{viral quasi-species genomes}~\cite{DBLP:conf/recomb/BaaijensSS20,BaaijensRKSS19}).
Each flow path represents a reconstructed sequence, and so a different set of flow
paths encodes a different set biological sequences, which may differ from the real
ones. If there are
multiple optimal flow decomposition solutions, then the reconstructed sequences
may not match the original ones, and thus be incorrect.
While many popular multiassembly tools look for minimum flow decompositions,
Williams et al.~\cite{williams2021flow} analyzed an error-free transcript
dataset to find that  20\% of human genes admit multiple minimum flow decomposition solutions.

\subsection{Safety Framework for Addressing Multiple Solutions}

Motivated by such an RNA assembly application,  Ma et al.~\cite{DBLP:conf/wabi/MaZK20} were the first to address the issue of multiple solutions to the flow decomposition problem under a probabilistic framework. Later, they~\cite{findingranges} solve a problem (\emph{AND-Quant}), which, in particular, leads to a quadratic-time algorithm for the following problem: given a flow in a DAG, and edges $e_1,e_2,\dots,e_k$, decide if in \emph{every} flow decomposition there is always a decomposed flow path passing through all of $e_1,e_2,\dots,e_k$. Thus, by taking the edges $e_1,e_2,\dots,e_k$ to be the edges of a path $P$, the AND-Quant problem can decide if a path $P$ (i.e., a given biological sequence) appears in all flow decompositions. This indicates that $P$ is likely part of some original RNA transcript.

We build upon the AND-Quant problem, by addressing the flow decomposition problem under the \emph{safety} framework~\cite{tomescu2017safe}, first introduced for genome assembly. For a problem admitting multiple solutions, a partial solution is said to be \emph{safe} if it appears in all solutions to the problem. For example, a path $P$ is safe for the flow decomposition problem, if for \emph{every} flow decomposition into paths ${\cal P}$, it holds that $P$ is a subpath of some path in $\cal P$. 
Further, 
a path $P$ is called \emph{$w$-safe} if in \textit{every} flow decomposition, $P$ is a subpath of some weighted path(s) in ${\cal P}$ whose total weight is at least $w$. 
In bioinformatics applications, it is common~\cite{tomescu2013novel,pertea2015stringtie,shao2017theory,kloster2018practical}
to look for a minimum cardinality path decomposition (or path cover, as in the case of \cite{trapnell2010transcript,liu2017strawberry}).
In this paper, we will consider \emph{any} flow decomposition as a valid solution, not only the ones of minimum cardinality. Dropping the minimality criterion is motivated by both theory and practice. On the one hand, since finding one minimum-cardinality flow decomposition is NP-hard~\cite{vatinlen2008simple}, we believe that finding all safe paths for them is also intractable. On the other hand, given the various issues with sequencing data, practical methods usually incorporate different variations of the minimum-cardinality criterion~\cite{bernard2013flipflop,DBLP:conf/recomb/BaaijensSS20,BaaijensRKSS19}. Thus, safe paths for \emph{all} flow decompositions are likely correct for many practical variations of the flow decomposition problem. 

Safety has precursors in combinatorial optimization, under the name of \emph{persistency}. For example, persistent edges present in all maximum bipartite matchings were studied by Costa~\cite{Costa1994143}. Incidentally, persistency has also been studied for the maximum flow problem, by finding persistent edges always having a non-zero flow value in any maximum flow solution (\cite{CECHLAROVA2001121} acknowledges~\cite{Lacko:1998aa} for first addressing the problem),
which is easily verified if the maximum flow decreases after removing the corresponding edge. 

In bioinformatics, safety has been previously studied for the genome assembly problem which at its core solves the problem of computing arc-covering walks on the assembly graph. Again since the problem admits multiple solutions where only one is correct, practical genome assemblers output only those solutions likely to be correct. The prominent approach dating back to 1995~\cite{KM95} is to compute trivially correct \emph{unitigs} (having internal nodes with \emph{unit} indegree and unit outdegree), which can be computed in linear time. Unitigs were generalised first in~\cite{PTW01}, and later~\cite{MGMB07,jacksonthesis,kingsford10} to be {\em extended} by adding their unique incoming and outgoing paths. These {\em extended unitigs}, though safe,
are not guaranteed to report \textit{everything} that can be correctly assembled,
presenting an important open question~\cite{Guenoche92,boisvert2010ray,nagarajan2009parametric,ShomoronyKCT16,LamKT14,BBT13} about the {\em assembly limit} (if any). This was finally resolved by Tomescu and Medvedev~\cite{tomescu2017safe} for a specific genome assembly formulation (single circular walk) by introducing {\em safe and complete} algorithms, which report everything that is theoretically reported as safe. Its running time was later optimized in~\cite{CairoMART19} and~\cite{cairo2020macrotigs}. Safe and complete algorithms were also studied by Acosta et al.~\cite{acosta2018safe} under a different genome assembly formulation of multiple circular walks. Recently, C\'aceres et al.~\cite{caceres2021safety} studied safe and complete algorithms for path covers in an application on RNA Assembly. 

\subsection{Safety in Flow Decomposition for RNA Assembly}

The prominent application of flow decomposition in bioinformatics is RNA transcript assembly, which is described as follows.
In complex
organisms, a gene may produce multiple RNA molecules
(\emph{RNA transcripts}, i.e., strings over an alphabet
of four characters), each having a different abundance. Currently, given a
sample, one can partially read the RNA transcripts and find their abundances using
\emph{high-throughput sequencing}~\cite{citeulike:3614773}. This technology
produces short overlapping substrings of the RNA transcripts. The main approach
for recovering the RNA transcripts from such data is to build an edge-weighted
DAG from these fragments, then to transform the weights into flow values by
various optimization criteria, and finally to decompose the resulting flow into an
``optimal'' set of weighted paths (i.e., the RNA transcripts and their
abundances in the sample)~\cite{DBLP:books/cu/MBCT2015}. A common strategy
for choosing the optimal set of weighted paths is to look for the parsimonious
solution, i.e., the solution with the fewest paths. Since this problem is NP-hard,
in practice many tools use the popular {\em greedy-width} heuristic \cite{tomescu2013novel,pertea2015stringtie}.
Greedy-width is also used in the assemblers for the related problem of viral quasispecies assembly~\cite{BaaijensRKSS19}. 
Further, some tools attempt to incorporate
additional information into the flow decomposition process, such as by using longer reads
or super reads~\cite{pertea2015stringtie,shao2017accurate,williams2021flow,gatter2019ryuto}.
Despite the large number of tools and methods that have been developed for
RNA transcript assembly, there is no method that consistently reports the correct
set of transcripts \cite{pertea2015stringtie,yu2020transborrow}. This suggests that the addressing the problem under the safety framework may be a promising approach. However, while a safe and complete solution clearly gives the maximally reportable correct solution,
it is significant to evaluate whether such a solution covers a large part of the true transcript, to be useful in practice. A possible application of such partial and  reliable solution is to consider them as constrains (see e.g. \mbox{\cite{williams2021flow}}) of real RNA transcript assemblers, to guide the assembly process of such heuristics. Another possible application could be to evaluate the accuracy of assemblers: does the output of the assembler include the safe and complete solution?.

\subsection{Our Results}
Our contributions can be succinctly described as follows.
\begin{enumerate}
\item \textbf{A simple local characterization resulting in an optimal verification algorithm}: 
We give a characterization for a safe path $P$ using its local property called {\em excess flow}.

\begin{restatable}{theorem}{flowSafety}
\label{thm:flowSafety}
For $w>0$, a path $P$ is $w$-safe iff its excess flow $f_P\geq w$.
\end{restatable}

The previous work~\cite{findingranges} on AND-Quant describes a {\em global} characterization using the maximum flow of the entire graph transformed according to $P$, requiring $O(mn)$ time. Instead, the excess flow is a {\em local} property of $P$ which is thus computable in time linear in the length of $P$. This also directly gives a simple verification algorithm which is optimal.

\begin{restatable}{theorem}{safePathVerification}
Given a flow graph (DAG)  having $n$ vertices and $m$ edges, it can be preprocessed in $O(m)$ time to verify the safety of any path $P$ in $O(|P|)=O(n)$  time.
\end{restatable}

\item \textbf{Simple enumeration algorithm}: 
The above characterization also results in a simple algorithm for reporting all maximal safe paths by using an arbitrary flow decomposition of the graph. 

\begin{restatable}{theorem}{safePathEnumeration}
Given a flow graph (DAG) having $n$ vertices and $m$ edges, all its maximal safe paths can be reported in $O(|{\cal P}_f|)=O(mn)$  time, where ${\cal P}_f$ is some flow decomposition. 
\end{restatable}

This approach starts with a candidate solution and uses the characterization on its subpaths in an efficient manner (a similar approach was previously used by~\cite{Costa1994143,acosta2018safe,caceres2021safety}).
The solution of the algorithm is reported using a compact representation (referred as ${\cal P}_c$), whose size can be   $\Omega(mn)$ is the worst case but merely $O(m+n)$ in the best case.\\ 
 
 \item \textbf{Empirically improved approach for RNA assembly:}
 Using simulated RNA splice graphs, we found that safe and complete paths for flow decomposition provide precise RNA assemblies while covering most of RNA transcripts. Safe and complete paths are $\approx 50\%$ better in coverage over previous notions of safe paths, while maintaining the perfect precision ensured by safety.
 Further, for the combined metric for coverage and precision (F-Score), the safe and complete paths outperform the popularly used greedy-width heuristic significantly ($\approx 20\%$)  and previous safety algorithms appreciably ($\approx 13\%$). Finally, though our approach takes $1.2-2.5\times$ time than the previous safety algorithms requiring equivalent memory, the greedy-width approach requires roughly $3-5\times$ time and $1.2-2.2\times$ memory than our approach. 
Hence, the significance of our approach in quality parameters increases with the increase in complex graph instances in the dataset, while the performance parameters are significantly better than greedy-width, without losing a lot over the previous safe algorithms.

\end{enumerate}

\section{Preliminaries and Notations} \label{prelim}
We consider a DAG $G=(V,E)$ with $n$ vertices and $m$ edges, where each edge $e$ has a positive flow $f(e)$ passing through it (also called its {\em weight}). Without loss of generality, we assume the graph is connected, and hence $m\geq n$.
For each vertex $u$, $f_{in}(u)$ and $f_{out}(u)$ denote the total flow on its incoming edges and outgoing edges, respectively. A vertex $v$ in the graph is called a {\em source} if $f_{in}(v)=0$ and a {\em sink} if $f_{out}(v)=0$. Every other vertex $v$ satisfies the {\em conservation of flow} $f_{in}(v)=f_{out}(v)$, making the graph a {\em flow graph}. 
For a path $P$ in the graph, $|P|$ denotes the number of its edges. For a set of paths ${\cal P}=\{P_1,\cdots,P_k\}$ we denote its total size (number of edges) by $|{\cal P}|=|P_1|+\cdots + |P_k|$. 

For any flow graph (DAG), its {\em flow decomposition} is a set of weighted {\em paths} ${\cal P}_f$ such that the flow on each edge of the flow graph equals the sum of the weights of the paths containing the edge. A flow decomposition of a graph can be computed in $O(|{\cal P}_f|)=O(mn)$ time using the simple path decomposition algorithm~\cite{AhukaMO93}.
A path $P$ is called \emph{$w$-safe} if, in every possible flow decomposition, $P$ is a subpath of some paths in ${\cal P}_f$ whose total weight is at least $w$. 
If $P$ is  $w$-safe with $w > 0$, we call $P$ a {\em safe flow path}, or simply {\em safe path}. Intuitively, for any edge $e$ with non-zero flow, we consider {\em where did the flow on $e$ come from?} We would like to report all the maximal paths ending with $e$ along which some $w > 0$ weight always ``flows'' to $e$ (see \Cref{fig:safeFlow}).
A safe path is {\em left maximal} (or {\em right maximal}) if extending it to the left (or right) with any edge makes it unsafe (i.e. not safe). A safe path is {\em maximal} if it is both left and right maximal. A set of safe paths is called {\em complete} if it consists of {\em all} the maximal safe paths.

\begin{figure}[t]
    \centering
\includegraphics[trim=5 15 20 5, clip, scale=1]{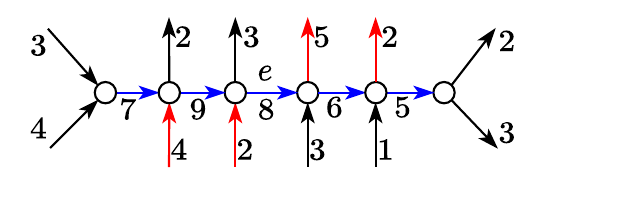}
    \caption{The prefix of the path (blue) up to $e$ contributes at least $2$ units of flow to $e$, as the rest may enter the path by the edges (red) with flow $4$ and $2$. Similarly, the suffix of the path (blue) from $e$ maintains at least $1$ unit of flow from $e$, as the rest may exit the path from the edges (red) with flow $5$ and $2$. Both these safe paths are {\em maximal} as they cannot be extended left or right.
    }
    \label{fig:safeFlow}
\end{figure}

Some previous notions used to describe safety for other problems also naturally extend to the flow decomposition problem as follows. The paths having internal nodes with unit indegree and unit outdegree are called {\em unitigs}~\cite{KM95}, which are trivially safe because every source-to-sink path which passes through an edge of unitig, also passes through the entire unitig contiguously. Further, a unitig can naturally be {\em extended} to include its unique incoming path (having nodes with unit indegree), and its unique outgoing path (having nodes with unit outdegree). This maximal extension of a unitig is called the {\em extended unitig}~\cite{PTW01,MGMB07,jacksonthesis,kingsford10}, which is also safe using the same argument. 

For some graphs the above notions already define the safety of flow decomposition {\em completely}. 
Recently, Millani et al.~\cite{millani2020efficient} defined a class of DAGs called 
\emph{funnels}, where every source-to-sink path is uniquely identifiable by at
least one  
edge, which is not used by any other source-to-sink path. Hence, considering such an edge as a trivial unitig (having a single edge), its extended unitig is exactly the corresponding source-to-sink path, making it safe. Thus, in a funnel  all the source-to-sink paths are naturally safe and hence trivially complete. Moreover, it implies that a  funnel has a unique flow decomposition, making the problem trivial for funnel instances. 

However, for non-funnel graphs unitigs and extended unitigs are safe but potentially not complete.
Note that  both unitigs and  extended unitigs are also safe for other problems dealing with unweighted graphs, such as path cover. Hence, they do not make use of the flows on the edges of the graph, potentially missing some paths that are safe for flow decomposition but not for unweighted problems like path cover.

\section{Characterization and Properties of Safe and Complete Paths}

The safety of a path can be characterized by its \emph{excess flow} (see \Cref{fig:excessP}), defined as follows.

\begin{definition}[Excess flow]
The \emph{excess flow} $f_P$ of a path $P=\{u_1,u_2,...,u_k\}$ is 
\begin{equation*}
f_P= f(u_1,u_2) - \sum_{\substack{u_i\in \{u_2,...,u_{k-1}\} \\ 
v\neq u_{i+1}}} f(u_i,v)
= f(u_{k-1},u_k) - \sum_{\substack{u_i\in \{u_2,...,u_{k-1}\} \\ 
v\neq u_{i-1}}} f(v,u_i) \end{equation*}
where the former and later equations are called {\em diverging} and {\em converging} formulations, respectively. \label{def:dispConv}
\end{definition}
\begin{remark} Alternatively, the converging and diverging formulations can be described as 
 \[f_P = \sum_{i=1}^{k-1} f(u_i,u_{i+1}) - \sum_{i= 2}^{k-1} f_{out}(u_i) =  \sum_{i=1}^{k-1} f(u_i,u_{i+1}) - \sum_{i=2}^{k-1} f_{in}(u_i).\]  
\label{obs:altCriteria}
\end{remark}

The converging and diverging formulations are equivalent by the conservation of flow on internal vertices. The idea behind the notion of an excess flow $f_P$ is that even in the worst case, the maximum {\em leakage} (see Figure~\ref{fig:excessP}), i.e., the flow leaving (or entering) $P$ at the internal nodes, is the sum of the flow on the outgoing (or incoming) edges of the internal nodes of $P$, that are not in $P$. However, if the value of incoming flow (or outgoing flow) is higher than this maximum leakage, then this excess value $f_P$ necessarily passes through the entire $P$.
The following results give the simple characterization and an additional property of safe paths.

\flowSafety*

\begin{proof}
The excess flow $f_P$ of a path $P$ trivially makes it $w\leq f_P$-safe by definition. 
If $f_P<w$, we can prove that $P$ is not $w$-safe by modifying any flow decomposition having $w$ flow on $P$ to leave only $f_P$ flow (or $0$, if $f_P<0$) on $P$ as follows. In \Cref{fig:excessP} (diverging), consider a flow path $P'$ entering $P$ through edge $e_1$ (except first edge (blue)) and leaving $P$ at an edge $e_2$ (red) except last edge of $P$. 
Since $f_P<w$, it is not possible that every path leaving $P$ using a red edge starts at the first blue edge (by definition of $f_P$), hence $P'$ always exists.
We modify $P'$ by using flow on $P$ to form two paths, which enter from $e_1$ and leave at the last edge, and which enter from the first edge and leave at $e_2$. We can repeat such modifications until flow on $P$ is $f_P$ (or $0$, if $f_P<0$) due to conservation of flow.
Additionally, for a path to be safe, it must hold that $w>0$. 
\end{proof}

\begin{figure}[t]
    \centering
\includegraphics[trim=5 7 10 7, clip,scale=.9]{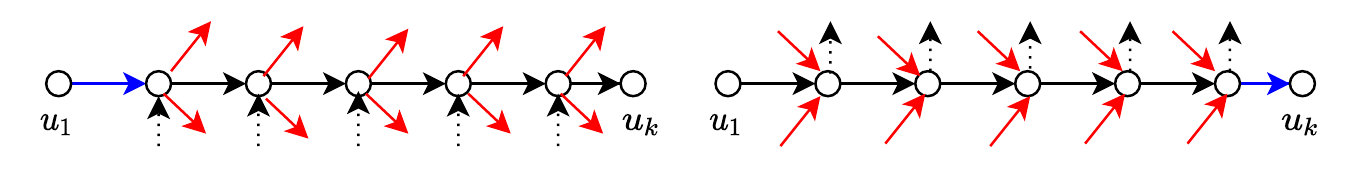}
    \caption{The excess flow of a path (left) is the incoming flow (blue) that necessarily passes through the whole path despite the flow (red) leaving the path at its internal vertices. It can be analogously described (right) as the incoming flow (blue) that necessarily passes through the whole path despite the flow (red) entering the path at its internal vertices.}
    \label{fig:excessP}
\end{figure}

\begin{lemma}
For any path in a flow graph (DAG), adding an edge $(u,v)$ to its start or its end, reduces its excess flow by
$f_{in}(v) - f(u,v)$, or $f_{out}(u)-f(u,v)$, respectively.
\label{lem:addEdge}
\end{lemma}

\begin{proof}
Using the converging formulation in \Cref{obs:altCriteria} adding an edge at the start of a path modifies its excess flow by $f(u,v)-f_{in}(v)$. Similarly, using the diverging formulation in \Cref{obs:altCriteria} adding an edge at the end of a path modifies its excess flow by $f(u,v)-f_{out}(u)$.
\end{proof}

\section{Simple Verification and Enumeration Algorithms} \label{verification_and_enumeration}
The characterization of a safe path in a flow graph can be directly adapted to simple algorithms for verification and enumeration of safe paths.

\subsection{Verification Algorithm}
The characterization  (\Cref{thm:flowSafety}) can be directly adapted to verify the safety of a path optimally. We preprocess the graph to compute the incoming flow $f_{in}(u)$ and outgoing flow $f_{out}(u)$ for each vertex $u$ in total $O(m)$ time. Using \Cref{obs:altCriteria} the time taken to verify the safety of any path $P$ is $O(|P|)=O(n)$, resulting in following theorem.

\safePathVerification*

\subsection{Enumeration of All Maximal Safe Paths}  
The maximal safe paths can be reported in $O(mn)$ time by  computing a candidate decomposition of the flow into paths, and verifying the safety of its subpaths using the characterization and a scan with the commonly used two-pointer approach.

The candidate flow decomposition can be computed in $O(mn)$ time using the classical flow decomposition algorithm~\cite{FordF10} resulting in $O(m)$ paths ${\cal P}_f$ each of $O(n)$ length. Now, we compute the maximal safe paths along each path $P\in {\cal P}_f$ by a two-pointer scan as follows. We start with the subpath containing the first two edges of the path $P$. We compute its excess flow $f$, and if $f>0$ we append the next edge to the path  on the right and incrementally compute its excess flow by \Cref{lem:addEdge}. Otherwise, if $f\leq 0$ we remove the first edge of the path from the left and incrementally compute the excess flow similarly by \Cref{lem:addEdge} (removing an edge $(u,v)$ would conversely modify the flow by $f_{in}(v)-f(u,v)$). We stop when the end of $P$ is reached with a positive excess flow.

The excess flow can be updated in $O(1)$ time when adding an edge to the subpath on the right or removing an edge from the left. 
If the excess flow of a subpath $P'$ is positive and on appending it with the next edge it ceases to be positive, we report $P'$ as a maximal safe path by reporting only its two indices on the path $P$. Thus, given a path of length $O(n)$, all its maximal safe paths can be reported in $O(n)$ time, and hence require total $O(mn)$ time for the $O(m)$ paths in the flow decomposition ${\cal P}_f$, resulting in the following theorem. 

\safePathEnumeration*

\paragraph*{Concise representation} The solution can be reported using a concise representation (referred as ${\cal P}_c$) having a set of paths as follows. We add to ${\cal P}_c$ every subpath of each path $P\in {\cal P}_f$  that contains maximal safe paths, along with the indices of the solution on the path. 
Thus, for one or more overlapping maximal safe subpaths from  $P$ we add a single path in ${\cal P}_c$ which is the union of all such maximal safe paths, making the paths added to ${\cal P}_c$ of {\em minimal} length. Finally, we also remove the duplicates and prefixes/suffixes among the maximal safe subpaths reported from different paths in ${\cal P}_f$ using an Aho Corasick Trie~\cite{AhoC75}, making the set of paths in ${\cal P}_c$ {\em minimal}. Thus, we define ${\cal P}_c$ as follows.

\begin{definition}[Concise representation ${\cal P}_c$] A minimal set of paths having a minimal length such that every safe path of the flow network is a subpath of some path in the set.
\end{definition}

\remark{In the worst case, the algorithm is optimal for DAGs having  $|{\cal P}_c|=|{\cal P}_f|=\Omega(mn)$, but in general $|{\cal P}_c|$ can be as small as $O(m+n)$ (see examples in \Cref{apn:wcSimple}). Thus, improving this bound requires us to not use a flow decomposition (and hence a candidate solution).}

\section{Experimental Evaluation}
We now evaluate the performance of our safe and complete algorithm by comparing it with the most promising algorithms for flow decomposition. 
Since the performance of various algorithms closely depend on the input graphs, we consider some practically relevant datasets to evaluate their true impact. As the most significant application of flow decomposition derives from RNA assembly, we consider the flow networks extracted as 
splice graphs of simulated RNA-Seq experiments. That is, starting from a set of RNA transcripts, we simulate their expression levels and superimpose the transcripts to create a flow graph. Evaluating our approach in such \emph{perfect} scenario allows us to remove the biases introduced by real RNA-Seq experiments~\cite{srivastava2020alignment} and focus the features offered by the each technique instead. Further, the performance of algorithms also closely depend on the complexity $k$ of a graph, that we measure as the number of paths in the ground truth decomposition of the graph. Thus, we present our results with regards to the complexity $k$ of the input graph instances. 

We first investigate the practical significance of {\em safety} by comparing our safe solution to a popularly used flow decomposition heuristic that is also scalable. The greedy-width~\cite{vatinlen2008simple} heuristic decomposes the flow by sequentially selecting the heaviest possible path, resulting in a simple algorithm that is both scalable and performs well in practice. However, any flow decomposition algorithm may not always report the ground truth paths that originally built the instance of the flow graph. Thus, it is important to measure the reported solution using a {\em precision} metric which evaluates the correctness of the solution. We thus investigate how the precision of greedy-width varies  particularly as the value of $k$ increases.

We then investigate the practical significance of {\em completeness} as reported by our solution, over the previously known safe solutions as reported by unitigs and extended unitigs (recall \Cref{prelim}). 
Note that every safe solution would always score perfectly in a precision metric by definition. Hence, all safe solutions would always outperform greedy-width (or any flow decomposition algorithm) in precision metrics. 
However, this perfect precision comes at the cost of the amount of the solution that is reported. 
Intuitively, this can be measured using some {\em coverage} metrics which describe how much of the ground truth sequence is included in the reported paths. Note that any flow decomposition algorithm will perform better than any safe algorithm by definition, as the safe paths are always subpaths of the paths reported by any flow decomposition algorithm.  Further, the extended unitigs would clearly outperform unitigs, and our safe paths would clearly outperform both unitigs and extended unitigs. We thus investigate how the coverage of various algorithms varies with respect to the greedy-width particularly as the value of $k$ increases. 

Finally, to understand the overall impact of different algorithms, where safe algorithms as compared to greedy-width clearly outperform in precision metrics and underperform in coverage metrics, we address both coverage and precision measures using a single metric, i.e., F-score. We thus investigate the variation in F-score over different values of $k$. In addition, to understand the practical utility of the algorithms we also investigate their performance measures in terms of running time and space requirements. 

\subsection{Datasets}

We consider two RNA transcripts datasets, generated based on approach of Shao et al. \cite{shao2017theory}. They created ``perfect" flow graphs where the true set of transcripts and abundances is always a flow decomposition of the graph (which also means that the graphs will satisfy conservation of flow). They start with a ground truth set of transcripts and abundances and create the input instances by superimposing them into a single graph, adding a single source
$s$ (and sink $t$) with an edge to the beginning (and end) of each transcript. 

\paragraph*{Funnel instances:}
As described in \cref{prelim}, in funnels~\cite{millani2020efficient} all paths are safe. This means that for any flow decomposition algorithm (including greedy-width) and most safe algorithms (including extended unitigs and our safe and complete algorithm), the resulting paths always achieve the perfect value of coverage, precision, and F-score on funnel instances. As a result, they dilute the relative measures of the different algorithms. Previously, Kloster et al.~ \cite[Lemma 8]{kloster2018practical} described a contraction of graphs that transforms funnels to trivial instances ($k=1$), however they excluded only single path instances from their evaluation. We found (see \Cref{fig:datasets}) that many complex instances (with larger $k$) are also funnels making them trivial. Hence, we removed such instances from our evaluation for a more accurate presentation of our results. Since  the previous studies~\cite{shao2017theory,kloster2018practical,williams2021flow} have considered the complete datasets including the trivial instances, we also include the evaluation on the complete datasets in \Cref{ap:funnels} for the sake of completeness. 

\begin{figure}[t]
     \centering
     \begin{subfigure}[b]{0.3\textwidth}
         \centering
         \includegraphics[width=0.9\textwidth]{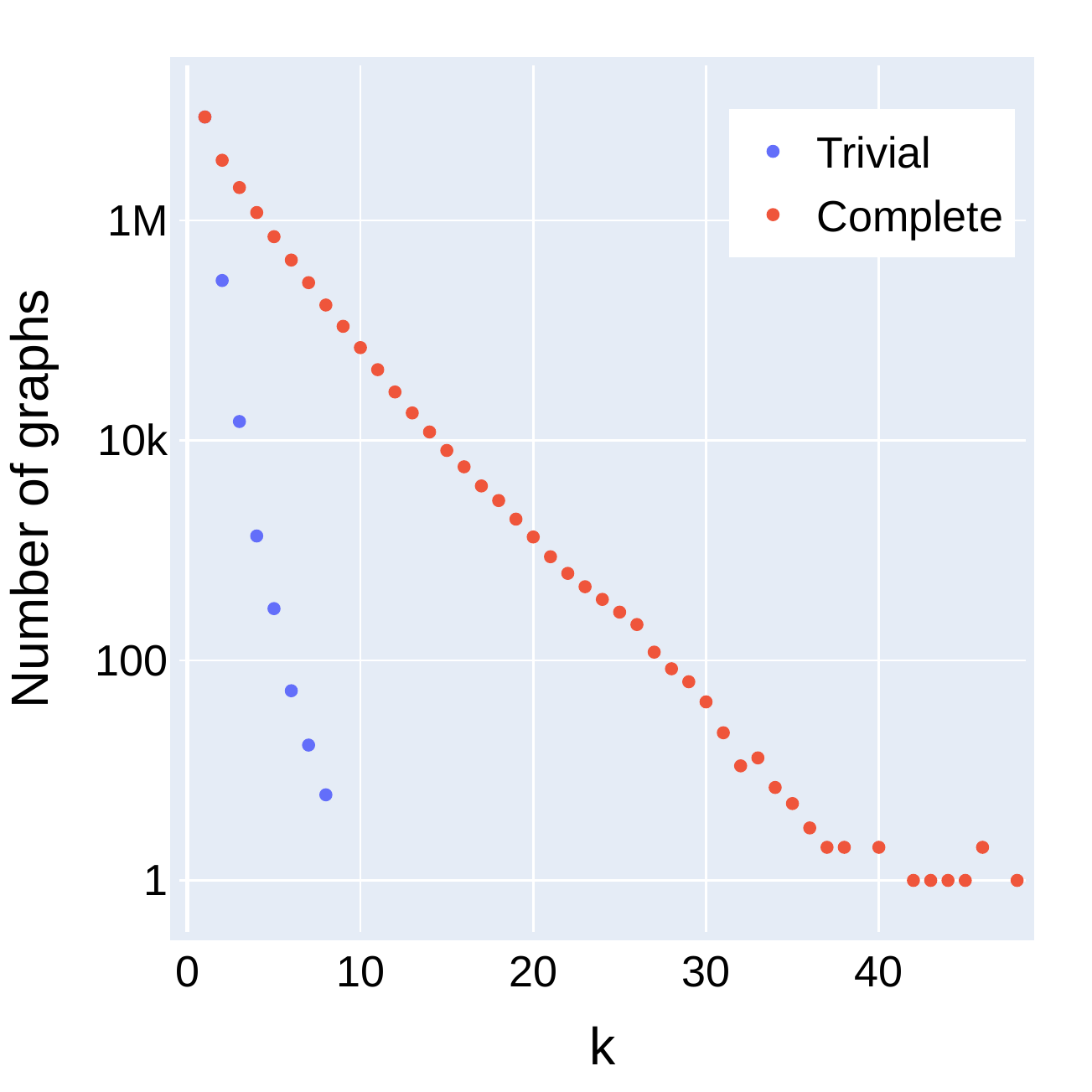}
         \caption{Catfish dataset}
        \label{fig:dataset1a}
     \end{subfigure}
    \begin{subfigure}[b]{0.3\textwidth}
         \centering
         \includegraphics[width=0.9\textwidth]{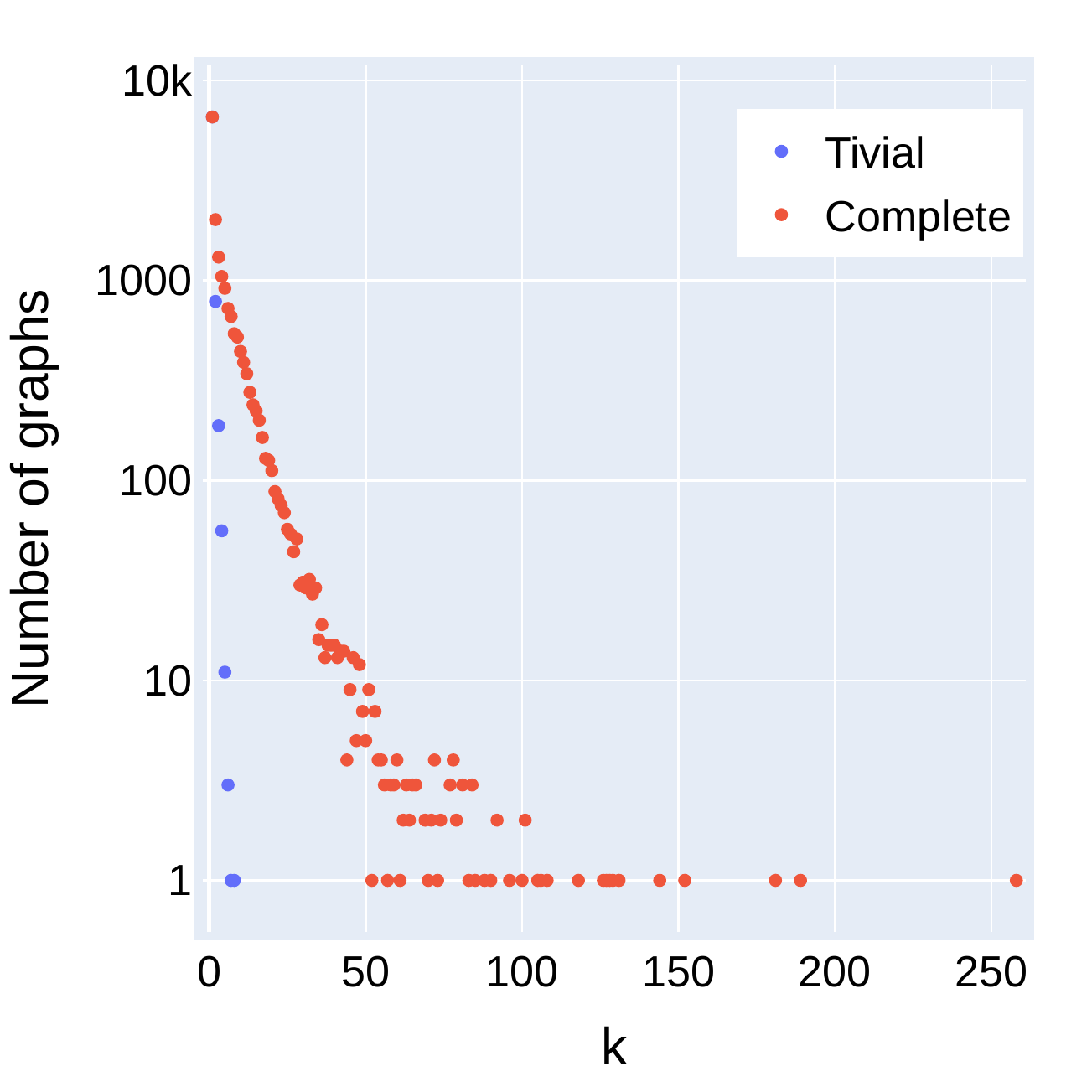}
         \caption{Reference-Sim dataset}
        \label{fig:dataset1b}
     \end{subfigure}
\caption{Distribution of graphs in the datasets by its complexity $k$ with respect to the trivial instances (funnels).}
\label{fig:datasets}
\end{figure}

\paragraph*{Catfish dataset:}
We consider the dataset first used by Shao and Kingsford~\cite{shao2017theory}, which includes 100 simulated human transcriptomes for human, mouse, and zebrafish using Flux-Simulator~\cite{griebel2012modelling}. Additionally, it includes 1,000
experiments from the Sequence Read Archive, with simulated abundances for transcripts using Salmon~\cite{patro2015salmon}.
In both cases,  the weighted transcripts are  superimposed to build splice graphs as described above. This dataset has also been used in other flow decomposition benchmarking studies~\cite{kloster2018practical,williams2021flow}. 
There are 17,335,407 graphs in total in this dataset, of which 8,301,682 are non-trivial (47.89\%). The logscale distribution of the complete dataset (and its funnels) by $k$ is shown in \Cref{fig:dataset1a}. 
However, in this dataset the details about the number of bases on each node (exons or pseudo-exons) are omitted, which results in an incomplete measure of coverage and precision. Moreover, this dataset has negligible complex graph instances (having large $k$). Hence, we do not include its evaluation in the main paper, rather defer it to \Cref{ap:catfish} for the sake of completeness.

\paragraph*{Reference-Sim dataset:} We consider a dataset~\cite{lucia_williams_2021_5646910} containing a single simulated transcriptome as follows. For each transcript on the positive strand in the GRCh.104 \emph{homo sapiens} reference genome, it samples a value from the lognormal distribution with mean and variance both equal to $-4$, as done in the default settings of RNASeqReadSimulator~\cite{li2014rnaseqreadsimulator}. It then multiplies the resulting values by 1000 and round to the nearest integer. Then it excludes any transcript that is rounded to 0. There are 17,941 total graphs in this dataset, of which 10,323 are non-trivial (57.54\%). The logscale distribution of the complete dataset (and its funnels) by $k$ is shown in \Cref{fig:dataset1b}. In this dataset, we also have access to the genomic coordinates (and hence number of bases) represented by nodes, allowing us to compute more practically relevant coverage and precision metrics.

\subsection{Evaluation Metrics}

All metrics are computed in terms of  bases for the Reference-Sim
dataset. However, since the Catfish dataset omits the base information its metrics are computed in terms of exons or pseudo-exons (vertices in
the flow graph). For every algorithm, $R$ denotes a reported path (for Catfish) or a reported safe subpath (for unitigs, extended unitigs, and safe complete) of the solution. In addition, $T$ denotes a path in the set of ground truth transcripts provided in the dataset. For each graph, we compute the following metrics which were also used earlier by~\cite{caceres2021safety} for safety in constrained path covers:
\setlist[description]{font=\normalfont\itshape\space}
\begin{description}
\item[Weighted precision:] We classify a reported path $R$ as \emph{correct} if $R$ is a subpath of some ground truth transcript $T$ of the flow graph. Weighted precision is the total length of correctly reported paths divided by the total length of reported paths. The commonly used precision metric~\cite{pertea2015stringtie,shao2017accurate} for measuring the accuracy of RNA assembly methods considers only those paths as correct which are (almost) exactly contained in the ground truth decomposition. Further, the precision is computed as the number of correctly reported paths divided by the total reported paths. However, since all the safe algorithms reports (possibly) partial transcripts, we use subpaths instead of (almost) exactly same paths. To highlight {\em how much} is reported correctly instead of {\em how many}, we use weighted precision to give a better score for longer correctly reported paths. 
    \item[Maximum relative coverage:] Given a ground truth transcript $T$ and a reported path $R$, we define a \emph{segment} of $R$ inside $T$ as a maximal subpath of $R$ that is also subpath of $T$.
  We define the maximum
    relative coverage of a ground truth transcript as the length of the longest segment of a reported path inside $T$, divided by the length of $T$. The corresponding value for the entire graph is the average of the values over all $T$. While it is common in the literature~\cite{pertea2015stringtie,shao2017accurate} to report \emph{sensitivity} (the proportion of ground truth transcripts that are correctly predicted), we measure correctness based
    on coverage since all the safe algorithms report paths that (possibly) do not cover an entire transcript. 
    \item[F-score:] The standard measure to combine precision and sensitivity is using an \emph{F-score}, which is the harmonic mean of the two. In our evaluation we correspondingly use the weighted precision and the maximum relative coverage for computing the F-score.

\end{description}

\subsection{Implementation and Environment Details}
We evaluate the following algorithms in our experiments.
\begin{description}
\item[Unitigs:] It computes the unitigs, by considering each unvisited edge in the topological order and extending it towards the right as long as the internal nodes have unit indegree and unit outdegree. The result ignores single edges.
\item[ExtUnitigs:] It computes the extended unitigs, by considering each unitig including single edges, and extending it towards the left as long as the internal nodes have unit indegree, and towards the right as long as internal nodes have unit outdegree.   
\item[Safe\&Comp:] It computes the safe and complete paths for flow decomposition using our enumeration algorithm described in \Cref{verification_and_enumeration}. However, since the metrics evaluation scripts uses  each safe path individually (as reported by other algorithms), we output all safe paths separately instead of using  ${\cal P}_c$. This increases the size of output and hence time complexity to $O(mn^2)$ from $O(mn)$ as stated in \Cref{thm:flowSafety}. 
\item[Greedy:] It computes the greedy-width heuristic using Catfish~\cite{shao2017theory} with the \texttt{-a greedy} parameter. 
\end{description}

All algorithms are implemented in C++, whereas the scripts for evaluating metrics are implemented in Python. The Unitigs, ExtUnitigs, and Safe\&Comp implementations use optimization level 3 of GNU C++ (compiled with $-O3$ flag), whereas the Greedy uses the optimizations of the Catfish pipeline. The Unitigs, ExtUnitigs, and Safe\&Comp  additionally require a post processing step using Aho Corasick Trie~\cite{AhoC75} for removing duplicates, and prefix/suffixes to make the set of safe paths minimal. However, the time and memory requirements are evaluated considering only the algorithm, and not post processing and metric evaluations which are not optimized. All performances were evaluated on a laptop using a single core (i5-8265U CPU \@ 1.60GHZ) having 15.3GB memory. The source code of our project is available on Github~\footnote{https://github.com/algbio/flow-decomposition-safety} under GNU Genral Public License v3 license. 

\subsection{Results} \label{sec:results}

We first evaluate the significance of {\em safety} among the reported solution. \Cref{fig:dataset_2_precision} compares the weighted precision of all the algorithms on the Reference-Sim dataset distributed over $k$. All the safe algorithms clearly report perfect precision as expected. However, the precision of the Greedy algorithm sharply declines with the increase in $k$, almost linearly to $30\%$ for $k=35$. This may be explained by the sharp increase in the number of possible paths in graphs with increase in $k$, which can be used by any flow decomposition algorithm. Hence, the significance of safety becomes very prominent as $k$ increases .

\begin{figure}[t]
     \centering
     \begin{subfigure}[b]{0.32\textwidth}
         \centering
         \includegraphics[trim=30 25 75 90,clip,  width=\textwidth]{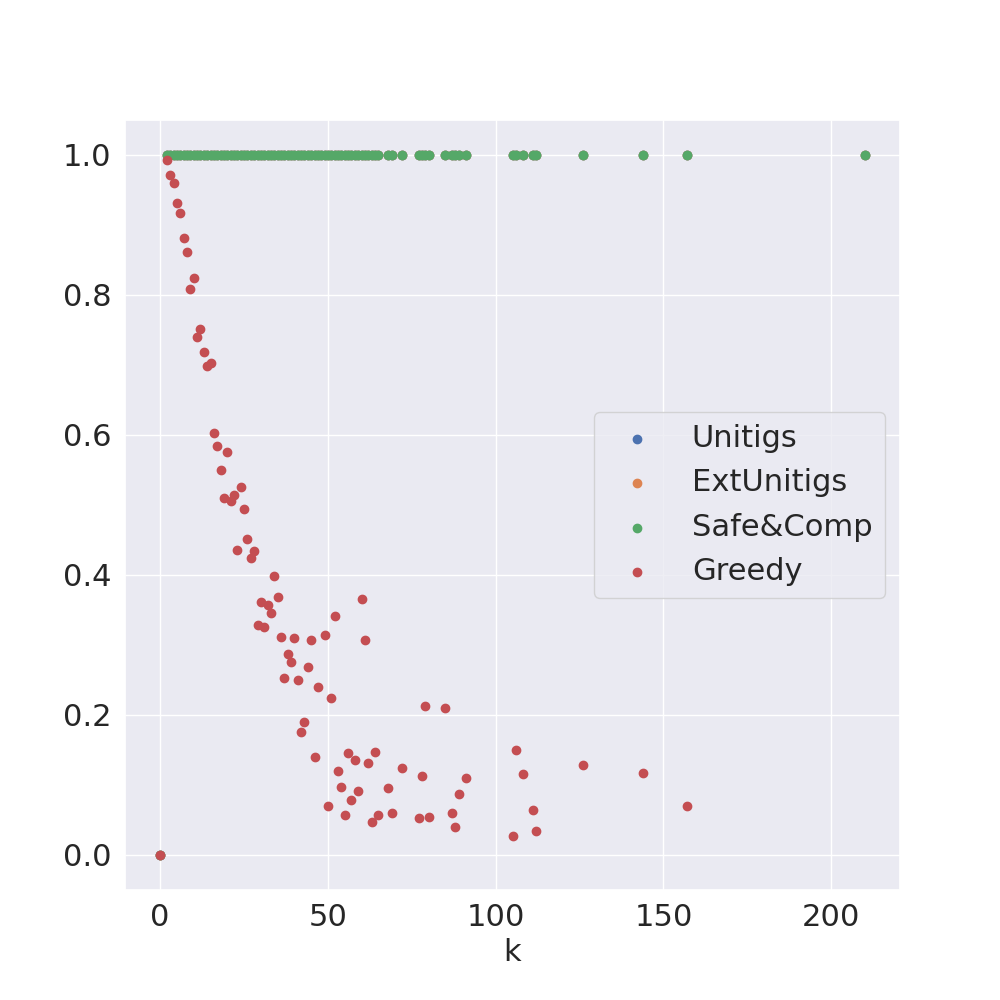}
         \caption{Weighted Precision}
          \label{fig:dataset_2_precision}
     \end{subfigure}
     \begin{subfigure}[b]{0.32\textwidth}
         \centering
         \includegraphics[trim=30 25 75 90,clip,  width=\textwidth]{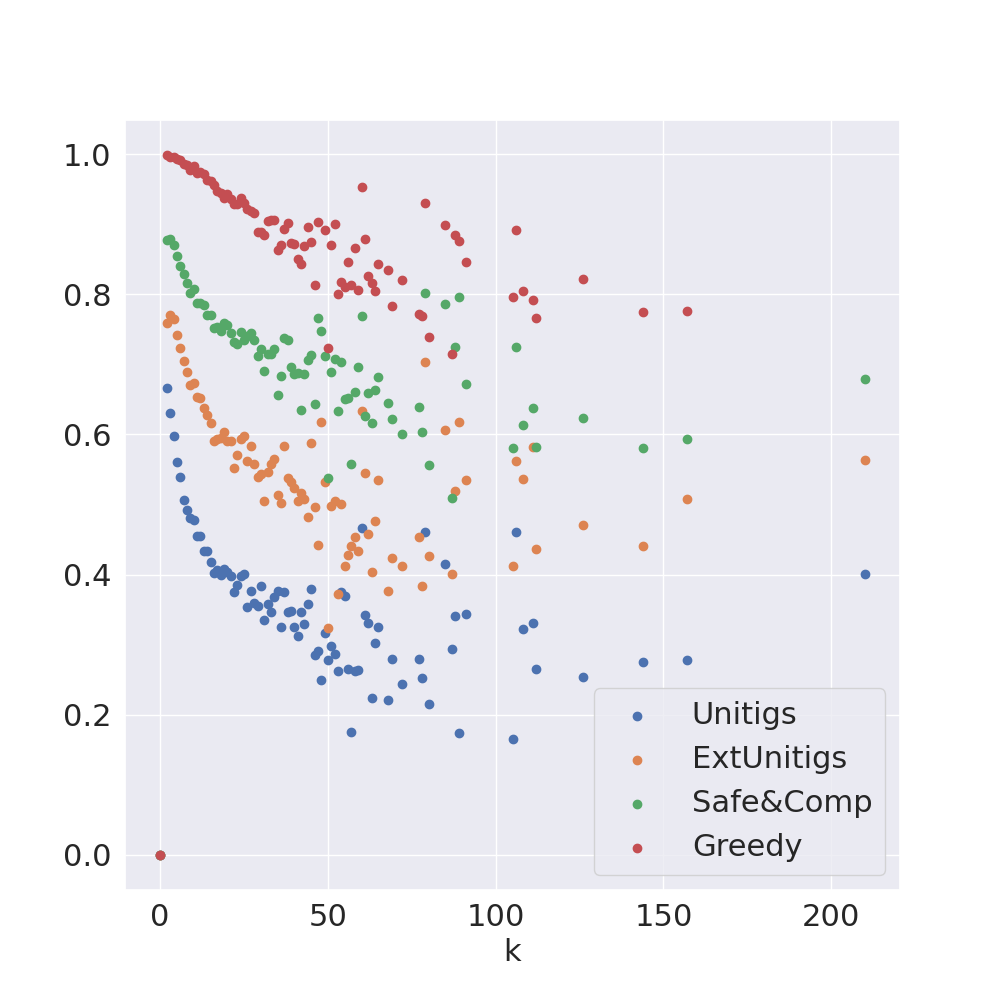}
         \caption{Maximum Relative Coverage}
         \label{fig:dataset_2_coverage}
     \end{subfigure}
    \begin{subfigure}[b]{0.32\textwidth}
         \centering
         \includegraphics[trim=30 25 75 90,clip,  width=\textwidth]{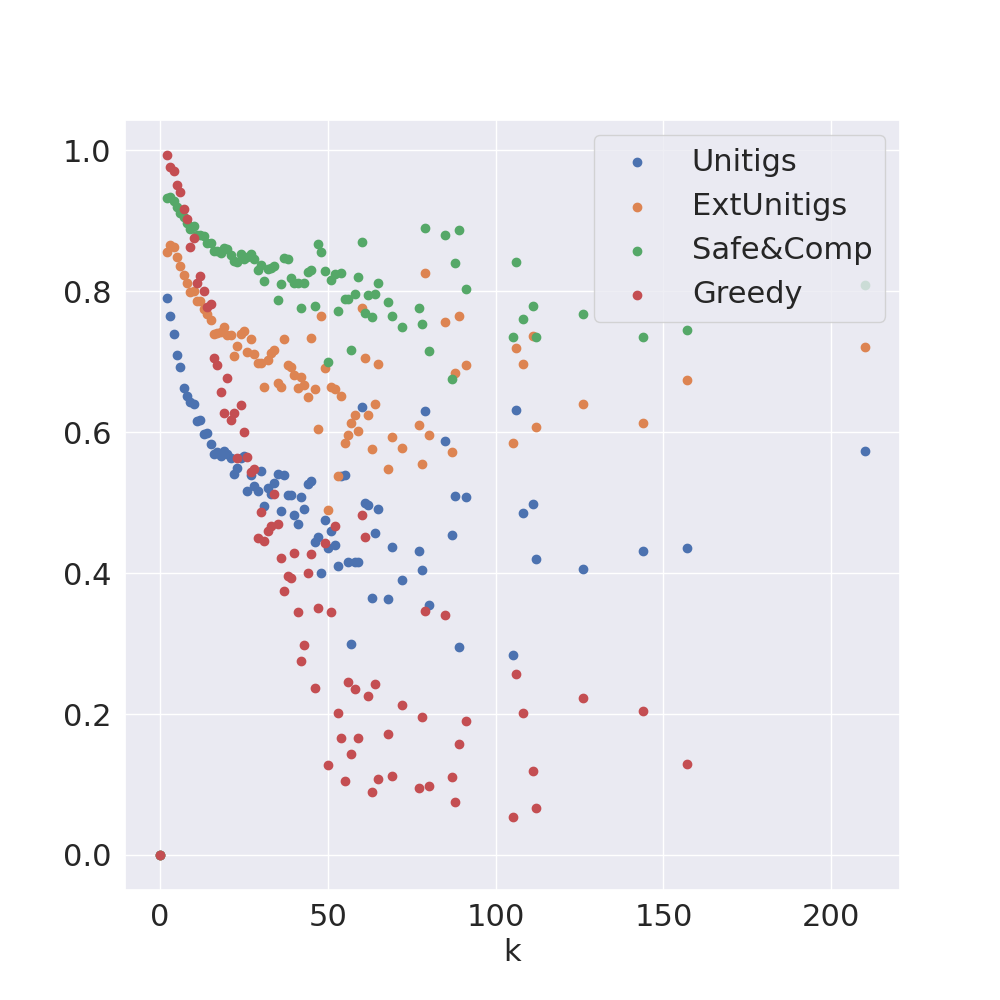}
         \caption{F-Score}
         \label{fig:dataset_2_fscore}
     \end{subfigure}
\caption{Evaluation metrics on graphs distributed by $k$ for the Reference-Sim dataset.}

        \label{fig:dataset_2_cov}
\end{figure}

We then evaluate the significance of {\em completeness} of the safe algorithms.  
\Cref{fig:dataset_2_coverage} compares the maximum relative coverage of all the algorithms on the Reference-Sim dataset distributed over $k$. As expected, Greedy outperforms all the other, followed by Safe\&Comp, ExtUnitigs and Unitigs. However, note that as $k$ reaches $20$ Safe\&Comp, ExtUnitigs and Unitigs sharply fall to $75\%$, $60\%$ and $40\%$,  while Greedy maintains around $95\%$ coverage. Overall, Safe\&Comp is almost always $\approx85-90$\% of that of Greedy, whereas ExtUnitigs and Unitigs falls to $60\%$ and $40\%$ respectively. Hence, the Safe\&Comp manages to maintain perfect precision without losing a lot on coverage, demonstrating the importance of {\em completeness} among the safe algorithms. 

\Cref{fig:dataset_2_fscore} supports the above inference by evaluating  the combined metric F-Score, where Safe\&Comp dominates Unitigs and ExtUnitigs by definition. Safe\&Comp also dominates Greedy as $k$ approaches $10$. It is also important to note that both ExtUnitigs and Unitigs eventually dominate Greedy for a slightly larger value of $k>20$ and $k>30$, respectively. This shows the significance of considering Safe algorithms for complex graphs. However, the significance of the Safe\&Comp is apparent from the \Cref{fig:dataset1b} as the number of graphs with such higher complexities also reduces sharply.

Hence, we evaluate a summary of the above results averaged over all graphs irrespective of $k$. \Cref{tab:annotated_noFunnels} summarizes the evaluation metrics for all the algorithms for simple graphs ($k<10$) and complex graphs ($k>10$), and both. While on the simpler graphs Greedy dominates Safe\&Comp mildly ($\approx 3\%$), for complex graphs it is dominated significantly ($\approx 20\%$) by Safe\&Comp and appreciably ($\approx 8\%$) by ExtUnitigs.  However, despite the larger ratio of simpler graphs, the collective F-score over all graphs is still ($\approx 4\%$) better for Safe\&Comp over Greedy which signifies the applicability of Safe\&Comp over Greedy.

\begin{table}[t]
\centering
\begin{tabular}{|c|c|c|c|c|c|c|}
\hline 
Graphs & Algorithm & Max. Coverage & Wt. Precision &  F-Score\\ \hline 
\multirow{3}{*}{\parbox[c]{2cm}{\begin{center}$k\geq 2$\\(100\%)\end{center}}} 
& Unitigs & 0.51 & 1.00 & 0.66\\
& ExtUnitigs & 0.69 & 1.00 & 0.81\\
& Safe\&Comp & 0.82 & 1.00 & 0.90\\
& Greedy & 0.98 & 0.81 & 0.86\\ \hline
\multirow{3}{*}{\parbox[c]{2cm}{\begin{center}$2\le k \le 10$\\(68\%)\end{center}}} & Unitigs & 0.55 & 1.00 & 0.70\\
& ExtUnitigs & 0.73 & 1.00 & 0.84\\
& Safe\&Comp & 0.84 & 1.00 & 0.91\\
& Greedy & 0.99 & 0.91 & 0.94\\
\hline
\multirow{3}{*}{\parbox[c]{2cm}{\begin{center}$k>10$\\(32\%)\end{center}}}  & Unitigs & 0.41 & 1.00 & 0.58\\
& ExtUnitigs & 0.61 & 1.00 & 0.75\\
& Safe\&Comp & 0.76 & 1.00 & 0.86\\
& Greedy & 0.95 & 0.60 & 0.69\\
\hline
\end{tabular}
\caption{Summary of evaluation metrics for the Reference-Sim dataset.}
\label{tab:annotated_noFunnels}
\end{table}

\begin{table}[t]
\centering
\begin{tabular}{|c|c|c|c|c|c|c|c|c|c|c|}
\hline 
\multirow{4}{*}{Algorithm} & \multicolumn{2}{|c|}{Reference-Sim} & \multicolumn{8}{|c|}{Catfish} \\ \cline{2-11} 
 & \multicolumn{2}{|c|}{Human} & \multicolumn{2}{|c|}{Zebrafish} & \multicolumn{2}{|c|}{Mouse} & \multicolumn{2}{|c|}{Human} & \multicolumn{2}{|c|}{Human (salmon)} \\ 
& \multicolumn{2}{|c|}{25.6MB} & \multicolumn{2}{|c|}{122MB} & \multicolumn{2}{|c|}{137MB} & \multicolumn{2}{|c|}{157MB} & \multicolumn{2}{|c|}{2.5GB} \\ 
\cline{2-11} 
 & Time & Mem & Time & Mem & Time & Mem & Time & Mem & Time & Mem \\ 
 \hline
Unitigs & 0.68s & 3.58MB & 13.82s & 3.51MB & 15.62s & 3.53MB & 18.22s & 3.54MB & 303.72s & 3.66MB \\
ExtUnitigs & 0.99s & 3.63MB & 18.31s & 3.52MB & 20.87s & 3.57MB & 23.64s & 3.56MB & 404.50s & 3.68MB \\
Safe\&Comp & 2.56s & 4.47MB & 20.17s & 3.56MB & 25.76s & 3.66MB & 28.59s & 3.54MB & 667.27s & 3.84MB \\
Greedy & 7.71s & 4.88MB & 108.30s & 6.00MB & 127.38s & 6.29MB & 148.46s & 6.34MB & 2684.30s & 8.47MB \\
\hline
\end{tabular}
\caption{Time and Memory requirements of the different algorithms for the evaluated datasets.}
\label{tab:timeMem}
\end{table}

Finally, we evaluate the applicability of the above algorithms in practice, by comparing their running time and peak memory requirements. Since all the algorithms are implemented in the same language (C++) and evaluated on the same machine, it is reasonable to directly compare these measures. 
In \cref{tab:timeMem}, we see that Unitigs clearly are the fastest, where ExtUnitigs takes roughly $1.3-1.5\times$  time. Safe\&Comp takes upto roughly $1.2-2.5\times$ time than ExtUnitigs, and Greedy requires roughly $3-5\times$ time than Safe\&Comp. The peak memory requirements of the safe algorithms are very close (within 5\%-25\%), whereas Greedy requires roughly $1.1-2.2\times$ more memory than Safe\&Comp. Overall, for the performance measures Safe\&Comp shows a significant improvement over Greedy, without losing a lot over the trivial algorithms.

\section{Conclusion}
We study the flow decomposition problem in DAGs under the Safe and Complete paradigm, which has applications in various domains, including the more prominent multi-assembly of biological sequences. Previous work characterized such paths (and their generalizations) using a global criterion. Instead, we present a simpler characterization based on a more efficiently computable local criterion, which is directly adapted into an optimal verification algorithm, and a simple enumeration algorithm. Intuitively, it is a {\em weighted} adaptation of {\em extended unitigs} which is a  prominent approach for computing safe paths. 

Through our experiments, we show that the safe and complete paths found by our algorithm outperform the popularly used greedy-width heuristic for RNA assembly instances with relatively complex graph instances, both on quality (F-score) and performance (running time and memory) parameters. On simple graphs, Greedy outperforms Safe\&Comp and Safe\&Comp outperforms ExtUnitigs mildly ($\approx 3-5\%$). However, on complex graphs, Safe\&Comp outperforms Greedy significantly ($\approx 20\%$) and ExtUnitigs appreciably ($\approx 13\%$). While the Reference-Sim dataset shows the overall dominance of Safe\&Comp since complex graphs are appreciable ($32\%$), Greedy dominates Safe\&Comp in Catfish dataset since complex graphs are negligible ($\approx 2\%$). Another significant reason for the dominance of Greedy over Safe\&Comp on Catfish datasets is the absence of base information on nodes (see \Cref{ap:catfish}). Hence, the importance of Safe\&Comp algorithms increases with the increase in complex graph instances in the dataset, and prominently when we consider information about the genetic information represented by each node. In terms of performance, ExtUnitigs are $1.3-1.5\times$ slower than the fastest approach (Unitigs), while Safe\&Comp further takes roughly $1.2-2.5\times$ time than ExtUnitigs, both requiring equivalent memory. However, Greedy requires roughly $3-5\times$ time and $1.1-2.2\times$ memory than Safe\&Comp. Overall, Safe\&Comp performs  significantly better than Greedy, without losing a lot over the trivial algorithms.

Despite the {\em optimality} of our characterization  of safe and complete paths, the enumeration algorithm is not time optimal. Additionally, the concise representation of the safe paths ${\cal P}_c$ may not be optimal for some graphs as described in \Cref{apn:wcSimple}. Hence, for datasets with more complex graphs there is a scope for improving the current enumeration algorithm and the concise representation in the future. Another interesting direction for an extension of this problem having practical significance is finding safe paths for those flow decompositions whose paths have a certain minimum weight threshold.

\bibliography{main}

\appendix

\section{Tightness and Worst Case for Simple Enumeration Algorithm}
\label{apn:wcSimple}

\begin{figure}[h]
    \centering
\includegraphics[trim=10 0 0 0, clip, scale=.9]{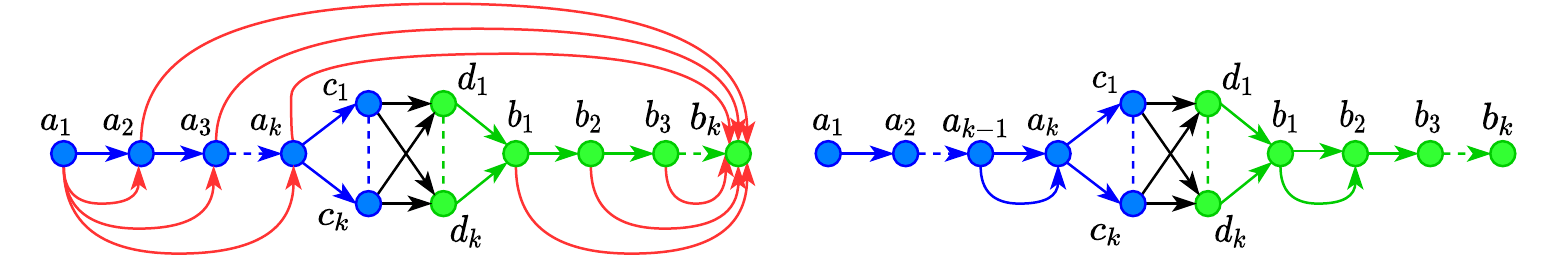}
    \caption{The worst case (left) and best case (right) graphs  for the simple enumeration algorithm.}
    \label{fig:wcSimple}
\end{figure}

The example shown in \Cref{fig:wcSimple} demonstrates the worst case and the best case graphs where the simple enumeration algorithm is optimal, and inefficient, respectively. We have two  paths $A=\{a_1,\cdots, a_k\}$ and $B=\{b_1,\cdots,b_k\}$. The set $C=\{c_1,\cdots, c_k\}$ has edges from $a_k$ and the set $D=\{d_1,\cdots,d_k\}$ has edges to $b_1$.  Choosing $k=n/4$ and any subset of connections between $C\times D$ we get a graph with any $n$ and $m$. Let there be flow $k$ on the {\em black} edges and unit flow on the {\em red} edges. (a) In the worst case graph (left) the flow on remaining edges is according to flow conservation assuming $a_1$ as the source and $b_k$ as the sink. Each edge in $C\times D$ necessarily has a separate path in ${\cal P}_f$ from $a_1$ to $b_k$, with $k$ maximal safe paths between $\{a_i,b_i\}$ for all $1\leq i\leq k$ because every path between $a_i$ to $b_1$ has excess flow $i$. This ensures that $|{\cal P}_c|=|{\cal P}_f|=\Omega(mn)$. (b) In the best case graph (right) the two edges from  $a_{k-1}$ to $a_k$ and from $b_1$ to $b_2$ carry equal flow, and the remaining edges have flow according to conservation of flow.  
    Each edge in $C\times D$ has a safe path of $O(1)$ size from $a_k$ to $b_1$. In addition there are two safe paths each of length $O(n)$ from $a_1$ to $a_k$, and from $b_1$ to $b_k$, corresponding to two parallel edges between $(a_{k-1},a_k)$, and between $(b_1,b_2)$, respectively. However, we have still have $|{\cal P}_f|=\Omega(mn)$ but $|{\cal P}_c|=O(m+n)$.

\section{Experimental Results on the Catfish Dataset} \label{ap:catfish}
In this section we evaluate the corresponding metrics of different algorithms on the Catfish dataset. However, since the Catfish dataset does not maintain the information on bases, the evaluation is  based on nodes  (exons or pseudo-exons) instead of bases.

\begin{figure}
     \centering
     \begin{subfigure}[b]{0.32\textwidth}
         \centering
         \includegraphics[trim=30 25 75 90,clip,  
         width=\textwidth]{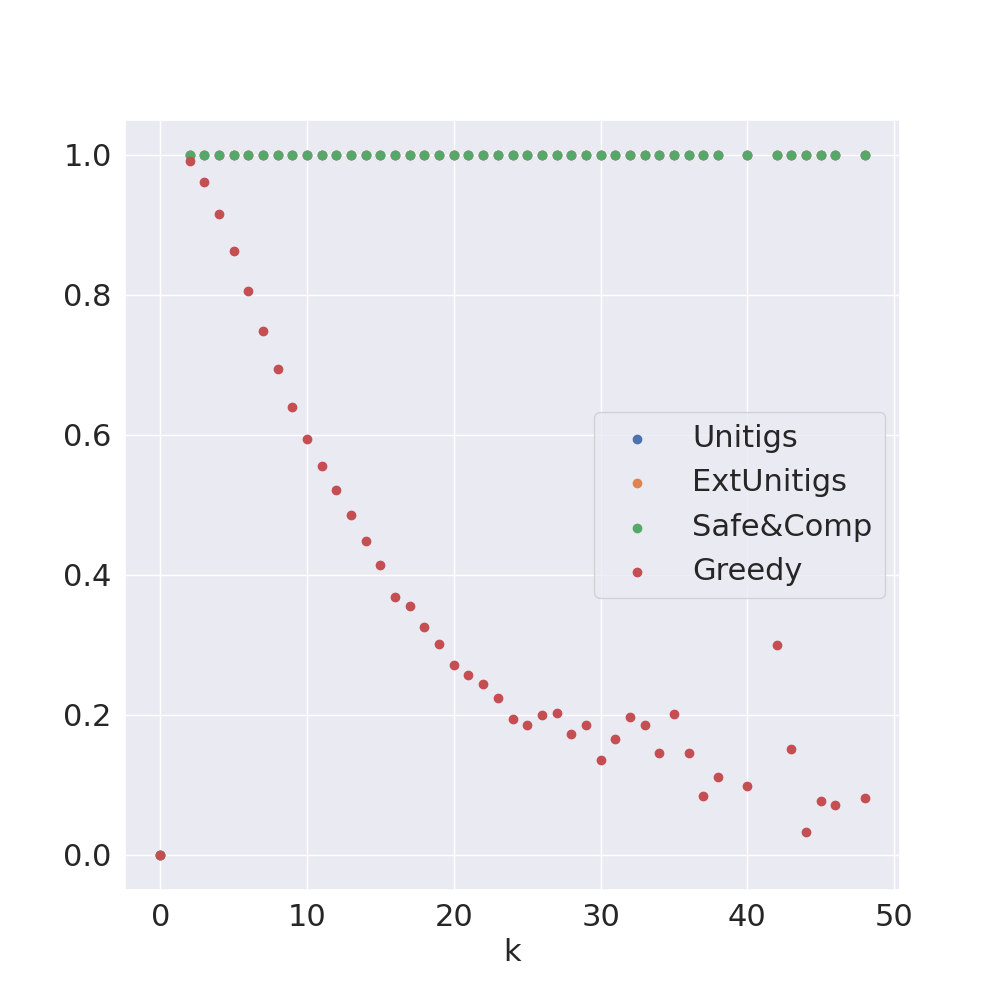}
         \caption{Weighted Precision}
          \label{fig:dataset_1_precision}
     \end{subfigure}
     \begin{subfigure}[b]{0.32\textwidth}
         \centering
         \includegraphics[trim=30 25 75 90,clip,  width=\textwidth]{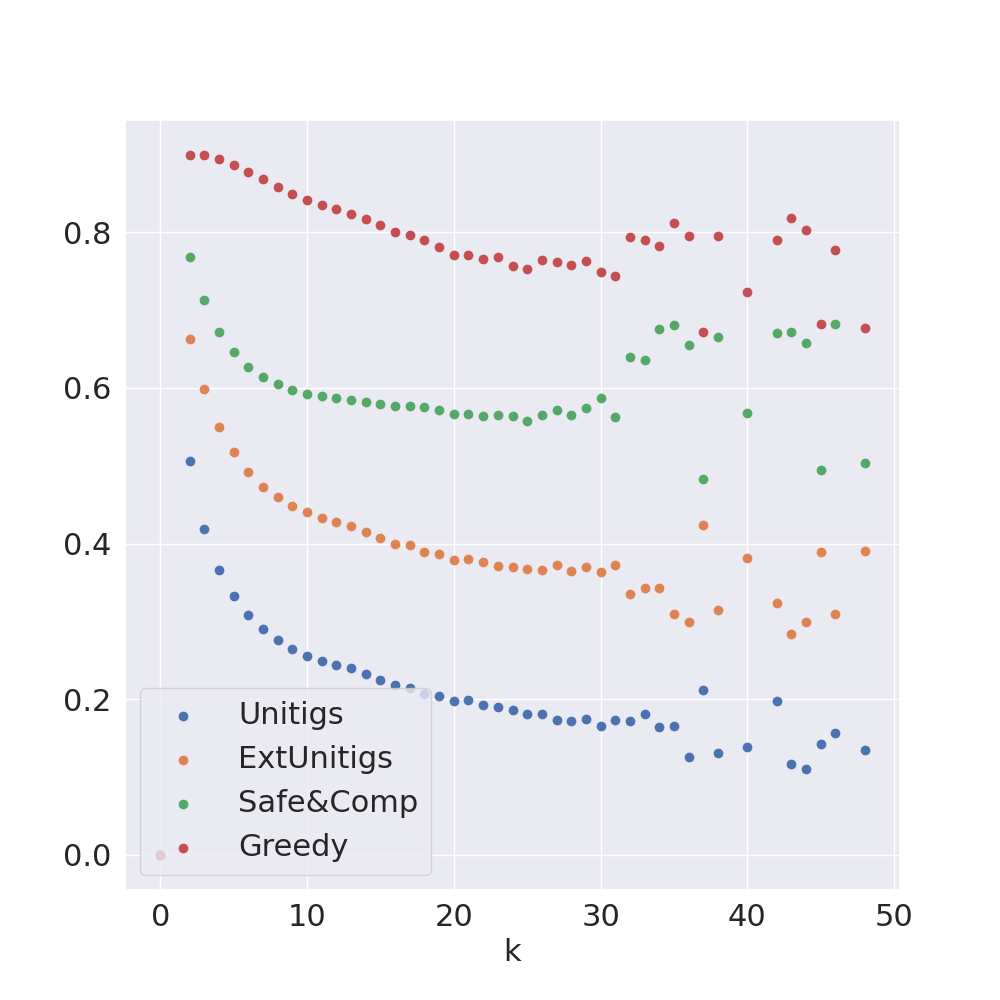}
         \caption{Maximum Relative Coverage}
         \label{fig:dataset_1_coverage}
     \end{subfigure}
    \begin{subfigure}[b]{0.32\textwidth}
         \centering
         \includegraphics[trim=30 25 75 90,clip,  width=\textwidth]{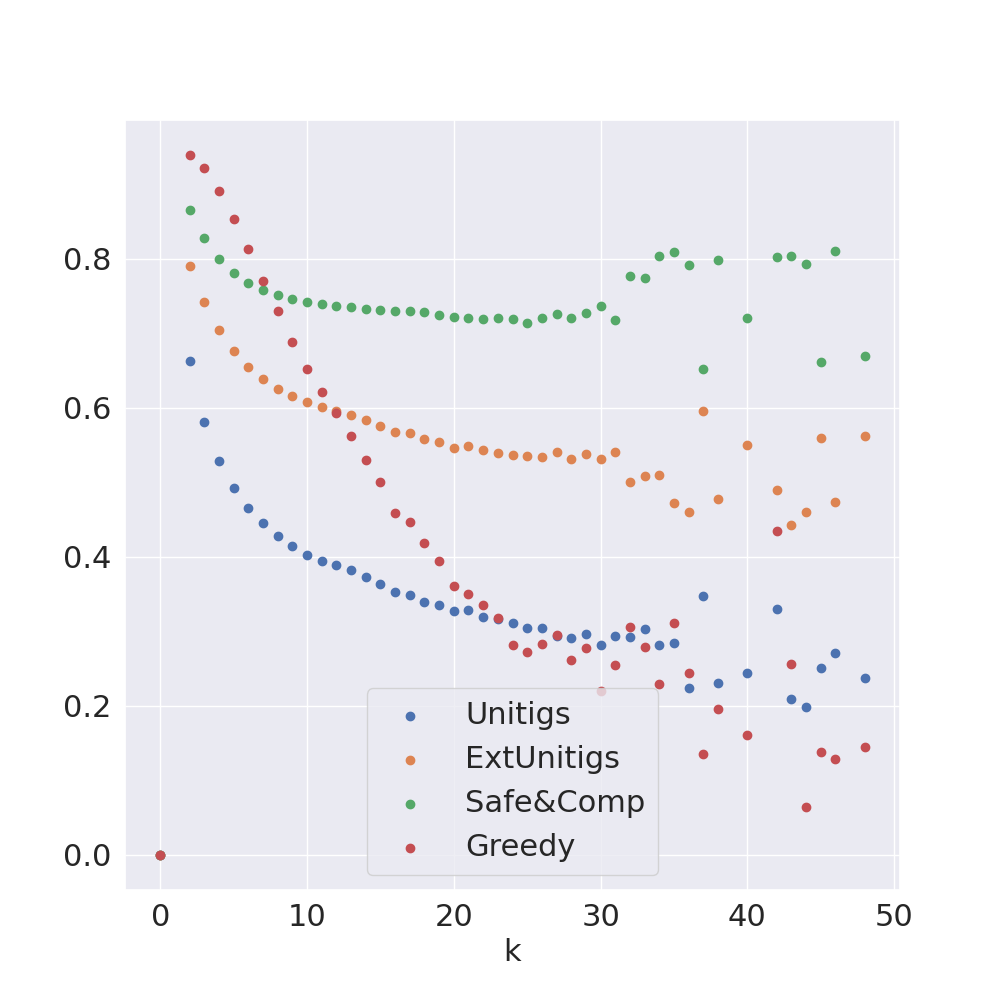}
         \caption{F-Score}
         \label{fig:dataset_1_fscore}
     \end{subfigure}
\caption{Evaluation metrics on graphs distributed by $k$ for the Catfish dataset.}

        \label{fig:dataset_1_cov}
\end{figure}

\begin{table}
\centering
\begin{tabular}{|c|c|c|c|c|c|c|}
\hline 
Graphs & Algorithm & Max. Coverage & Wt. Precision &  F-Score\\ \hline 
\multirow{3}{*}{\parbox[c]{2cm}{\begin{center}$k\geq 2$\\(100\%)\end{center}}} & Unitigs & 0.42 & 1.00 & 0.58\\
& ExtUnitigs & 0.59 & 1.00 & 0.73\\
& Safe\&Comp & 0.71 & 1.00 & 0.82\\
& Greedy & 0.89 & 0.92 & 0.89\\

\hline
\multirow{3}{*}{\parbox[c]{2cm}{\begin{center}$2\le k \le 10$\\(98\%)\end{center}}}  & Unitigs & 0.42 & 1.00 & 0.58\\
& ExtUnitigs & 0.59 & 1.00 & 0.74\\
& Safe\&Comp & 0.71 & 1.00 & 0.83\\
& Greedy & 0.89 & 0.93 & 0.90\\

\hline
\multirow{3}{*}{\parbox[c]{2cm}{\begin{center}$k>10$\\(2\%)\end{center}}} & Unitigs & 0.24 & 1.00 & 0.38\\
& ExtUnitigs & 0.42 & 1.00 & 0.59\\
& Safe\&Comp & 0.58 & 1.00 & 0.74\\
& Greedy & 0.82 & 0.49 & 0.56\\

\hline
\end{tabular}
\caption{Summary of evaluation metrics for the Catfish dataset, computed relative to nodes.}
\label{tab:catfish_nofunnels}
\end{table}

\begin{table}
\centering
\begin{tabular}{|c|c|c|c|c|c|c|}
\hline 
Species & Algorithm & Max. Coverage & Wt. Precision &  F-Score\\ \hline 
\multirow{3}{*}{\parbox[c]{2cm}{\begin{center} Human (Salmon) \end{center}}} & Unitigs & 0.41 & 1.00 & 0.57\\
& ExtUnitigs & 0.58 & 1.00 & 0.73\\
& Safe\&Comp & 0.70 & 1.00 & 0.82\\
& Greedy & 0.89 & 0.91 & 0.89\\

\hline
\multirow{3}{*}{\parbox[c]{2cm}{\begin{center} Human \end{center}}} & Unitigs & 0.44 & 1.00 & 0.60\\
& ExtUnitigs & 0.61 & 1.00 & 0.75\\
& Safe\&Comp & 0.73 & 1.00 & 0.84\\
& Greedy & 0.90 & 0.97 & 0.93\\

\hline
\multirow{3}{*}{\parbox[c]{2cm}{\begin{center} Zebrafish \end{center}}} & Unitigs & 0.49 & 1.00 & 0.65\\
& ExtUnitigs & 0.67 & 1.00 & 0.80\\
& Safe\&Comp & 0.78 & 1.00 & 0.87\\
& Greedy & 0.91 & 0.98 & 0.94\\

\hline
\multirow{3}{*}{\parbox[c]{2cm}{\begin{center} Mouse \end{center}}} & Unitigs & 0.46 & 1.00 & 0.61\\
& ExtUnitigs & 0.63 & 1.00 & 0.77\\
& Safe\&Comp & 0.75 & 1.00 & 0.85\\
& Greedy & 0.90 & 0.97 & 0.93\\

\hline
\end{tabular}
\caption{Summary of evaluation metrics for the Catfish dataset, computed relative to nodes, by species.}
\label{tab:catfish_nofunnels_species}
\end{table}

\begin{remark}
The results on the Catfish dataset do not match the inferences from \cref{sec:results} exactly. The primary differences and the expected reasons for the same are as follows: 
\begin{enumerate}
    \item Base vs. node computations for metrics: When considering the genomic content that is predicted (i.e., bases), Safe\&Comp outperforms Greedy with respect to F-score over all graphs, as seen in \cref{tab:annotated_noFunnels}. However, the opposite is observed when we compute metrics based on nodes, as reported in \cref{tab:human-nodes}. Because the Catfish dataset has no base information, we can only report node information, but it is possible that the same patterns we observe in Reference-Sim with bases would hold for Catfish in terms of bases as well.
    \item Ratio of simpler graphs: Catfish datasets are more skewed toward simpler graphs than the Reference-Sim dataset. \cref{tab:annotated_noFunnels} shows that Reference-Sim has 32\% of graphs with $k>10$, while \cref{tab:catfish_nofunnels} shows that Catfish dataset has only 2\%. Since Greedy outperforms Safe\&Comp on simpler graphs, it is better for overall Catfish Datasets having more simpler graphs. 
\end{enumerate}
\end{remark}

\section{Experimental Results on the Reference-Sim Dataset Considering Nodes} 
\label{ap:nodeRef}
In this section we evaluate the corresponding metrics of different algorithms on the Reference-Sim dataset based on nodes  (exons or pseudo-exons) instead of bases.

\begin{figure}
     \centering
     \begin{subfigure}[b]{0.32\textwidth}
         \centering
         \includegraphics[trim=30 25 75 90,clip,  
         width=\textwidth]{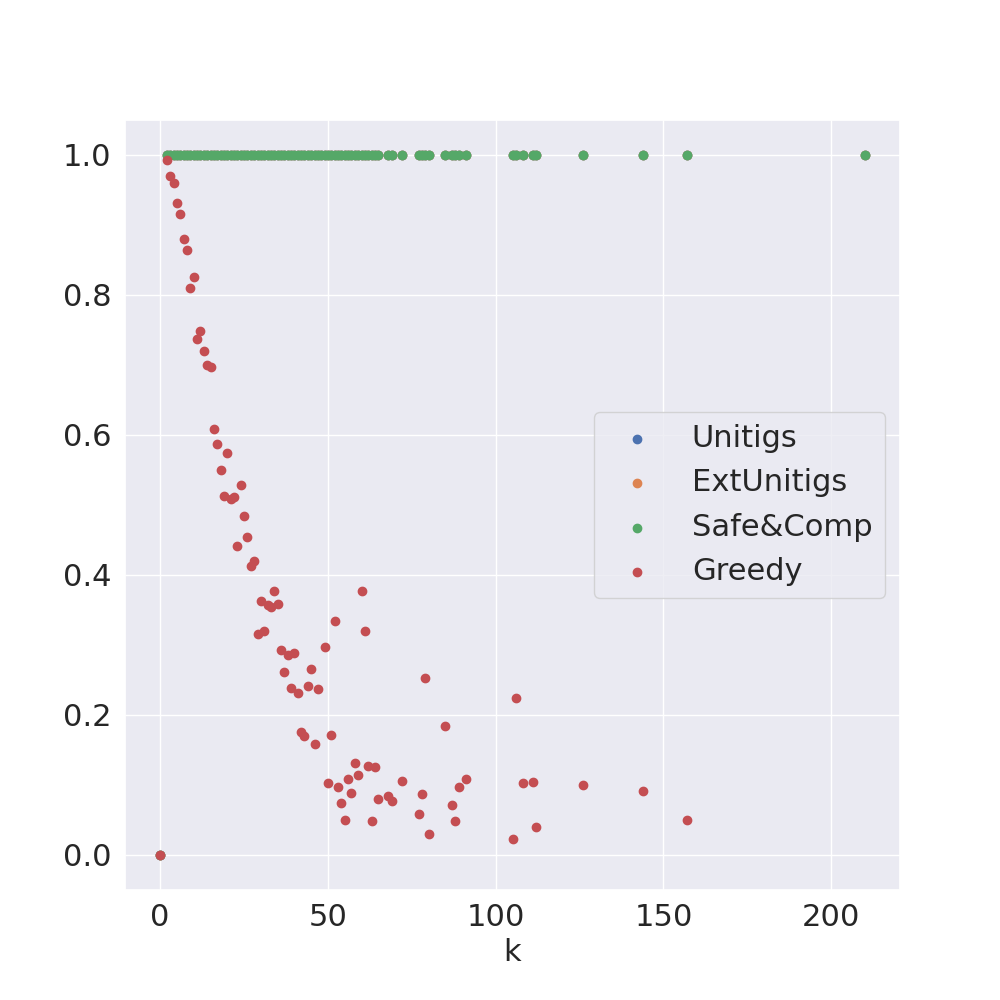}
         \caption{Weighted Precision}
          \label{fig:dataset_1_precision_node}
     \end{subfigure}
     \begin{subfigure}[b]{0.32\textwidth}
         \centering
         \includegraphics[trim=30 25 75 90,clip,  width=\textwidth]{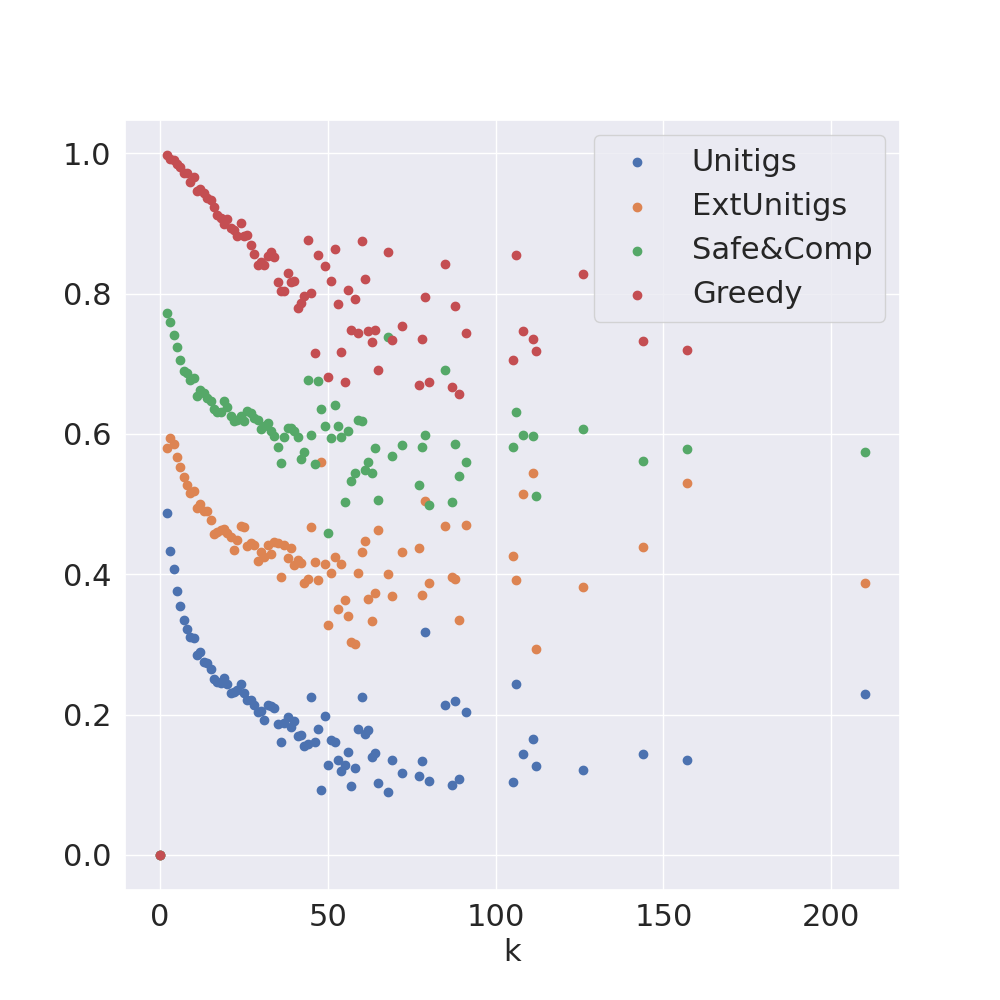}
         \caption{Maximum Relative Coverage}
         \label{fig:dataset_1_coverage_node}
     \end{subfigure}
    \begin{subfigure}[b]{0.32\textwidth}
         \centering
         \includegraphics[trim=30 25 75 90,clip,  width=\textwidth]{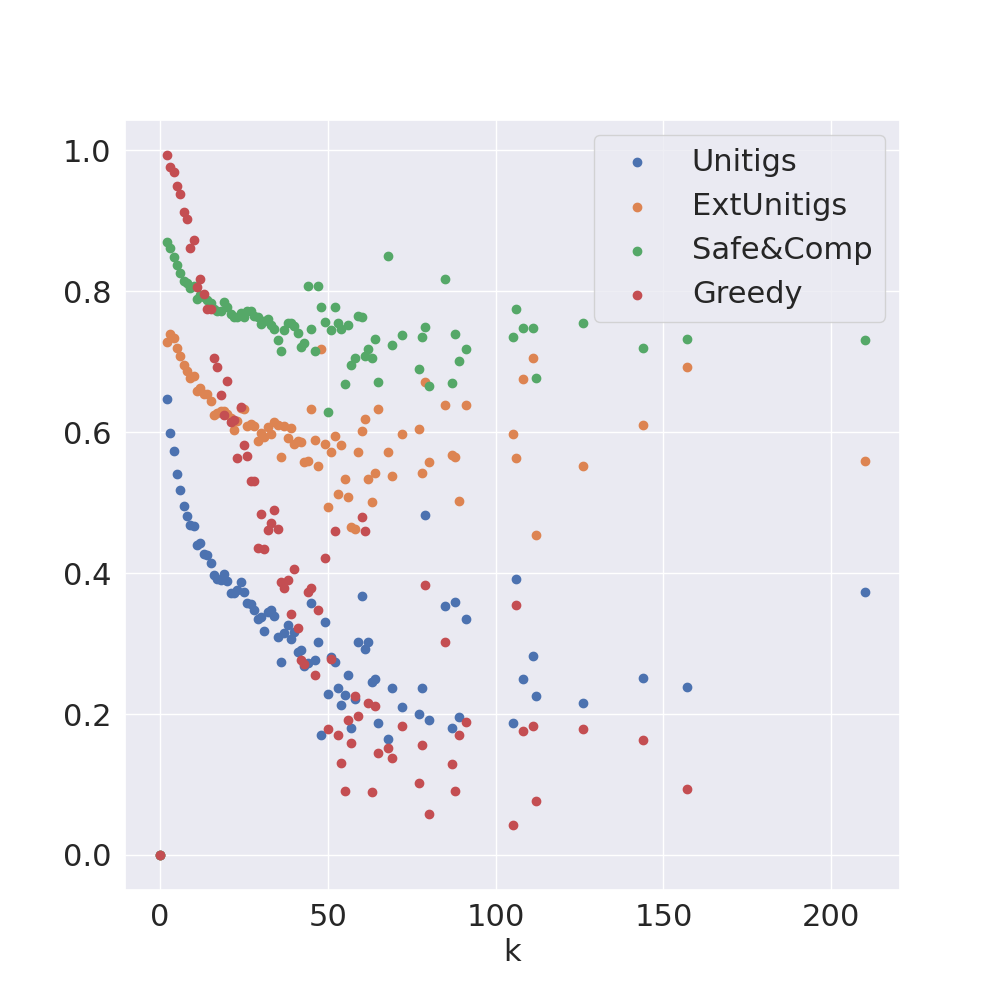}
         \caption{F-Score}
         \label{fig:dataset_1_fscore_node}
     \end{subfigure}
\caption{Evaluation metrics on graphs distributed by $k$ for the Reference-Sim dataset considering nodes.}

        \label{fig:dataset_1_cov_node}
\end{figure}

\begin{table}
\centering
\begin{tabular}{|c|c|c|c|c|c|c|}
\hline 
Graphs & Algorithm & Max. Coverage & Wt. Precision &  F-Score\\ \hline 
\multirow{3}{*}{\parbox[c]{2cm}{\begin{center}$k\geq 2$\\(100\%)\end{center}}} & Unitigs & 0.33 & 1.00 & 0.49\\
& ExtUnitigs & 0.53 & 1.00 & 0.69\\
& Safe\&Comp & 0.69 & 1.00 & 0.81\\
& Greedy & 0.96 & 0.81 & 0.86\\
\hline
\multirow{3}{*}{\parbox[c]{2cm}{\begin{center}$2\le k \le 10$\\(68\%)\end{center}}} & Unitigs & 0.37 & 1.00 & 0.53\\
& ExtUnitigs & 0.56 & 1.00 & 0.71\\
& Safe\&Comp & 0.72 & 1.00 & 0.83\\
& Greedy & 0.98 & 0.91 & 0.93\\
\hline
\multirow{3}{*}{\parbox[c]{2cm}{\begin{center}$k>10$\\(32\%)\end{center}}} & Unitigs & 0.25 & 1.00 & 0.40\\
& ExtUnitigs & 0.47 & 1.00 & 0.64\\
& Safe\&Comp & 0.64 & 1.00 & 0.78\\
& Greedy & 0.91 & 0.60 & 0.69\\
\hline
\end{tabular}
\caption{Summary of evaluation metrics for the Reference-Sim dataset, computed relative to nodes.}
\label{tab:human-nodes}
\end{table}

\section{Experimental Results Including Funnel Instances} \label{ap:funnels}
In this section we evaluate the corresponding metrics of different algorithms on the {\em complete}  Reference-Sim dataset and the {\em complete} Catfish dataset which include the funnels. 

\begin{remark}
The results when considering the complete datasets (including funnels) are diluted when compared to inferences from \cref{sec:results}. 
In this case, we expect that the differences between the algorithms become less sharp, because all algorithms solve trivial (funnel) instances perfectly, which artificially increases the precision and coverage scores. This is confirmed by comparing \cref{tab:annotated_noFunnels} to \cref{tab:human_funnels}. Without funnel instances, we observe overall F-scores range between 0.66 and 0.9; whereas the range is from 0.82 to 0.95 when including them. A similar effect occurs for Catfish data in \cref{tab:catfish_nofunnels} and \cref{tab:catfish_funnels}. 
This is also visible from coverage and F-score metrics in \Cref{fig:Funnel_dataset_2} and \Cref{fig:Funnel_dataset_1} which start from $100\%$ even for safe paths, which is not the case in corresponding figures without funnels. 
\end{remark}

\begin{figure}
     \centering
     \begin{subfigure}[b]{0.32\textwidth}
         \centering
         \includegraphics[trim=30 25 75 90,clip,  
         width=\textwidth]{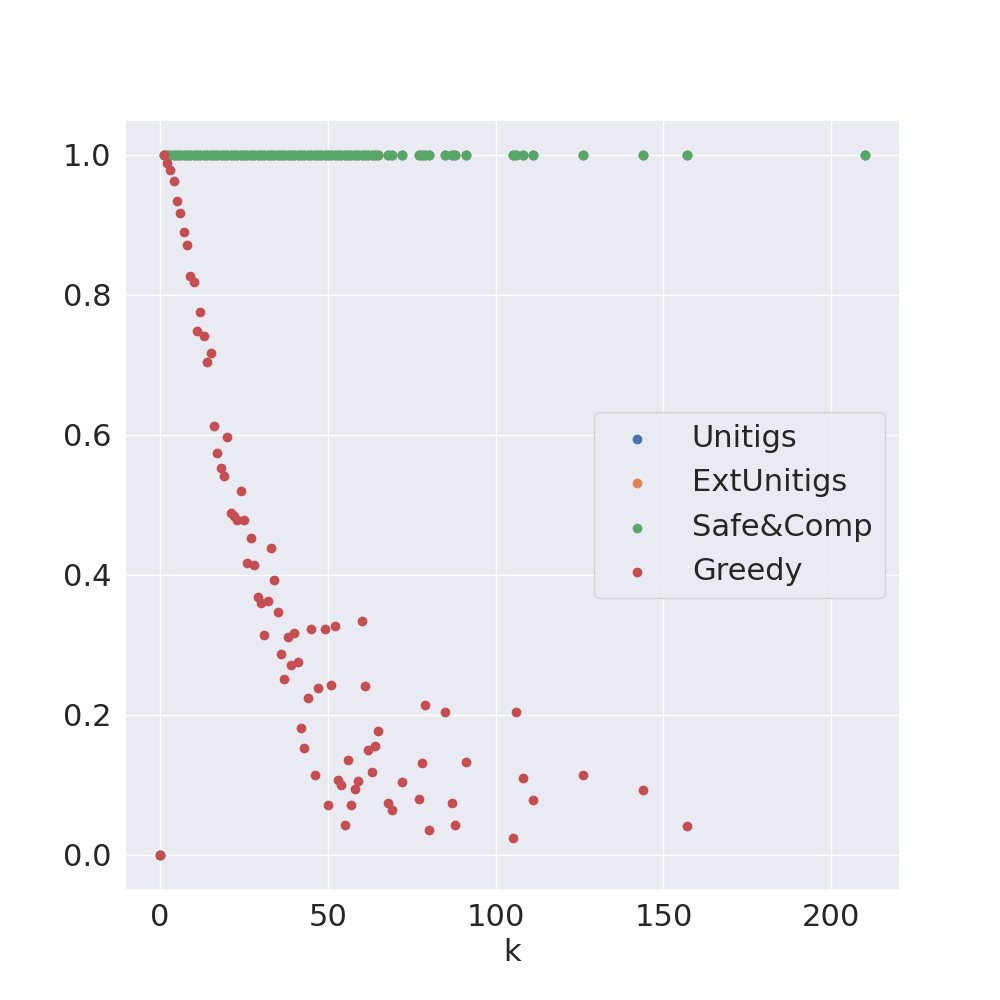}
         \caption{Weighted Precision}
          \label{fig:Funnel_dataset_2_precision}
     \end{subfigure}
     \begin{subfigure}[b]{0.32\textwidth}
         \centering
         \includegraphics[trim=30 25 75 90,clip,  width=\textwidth]{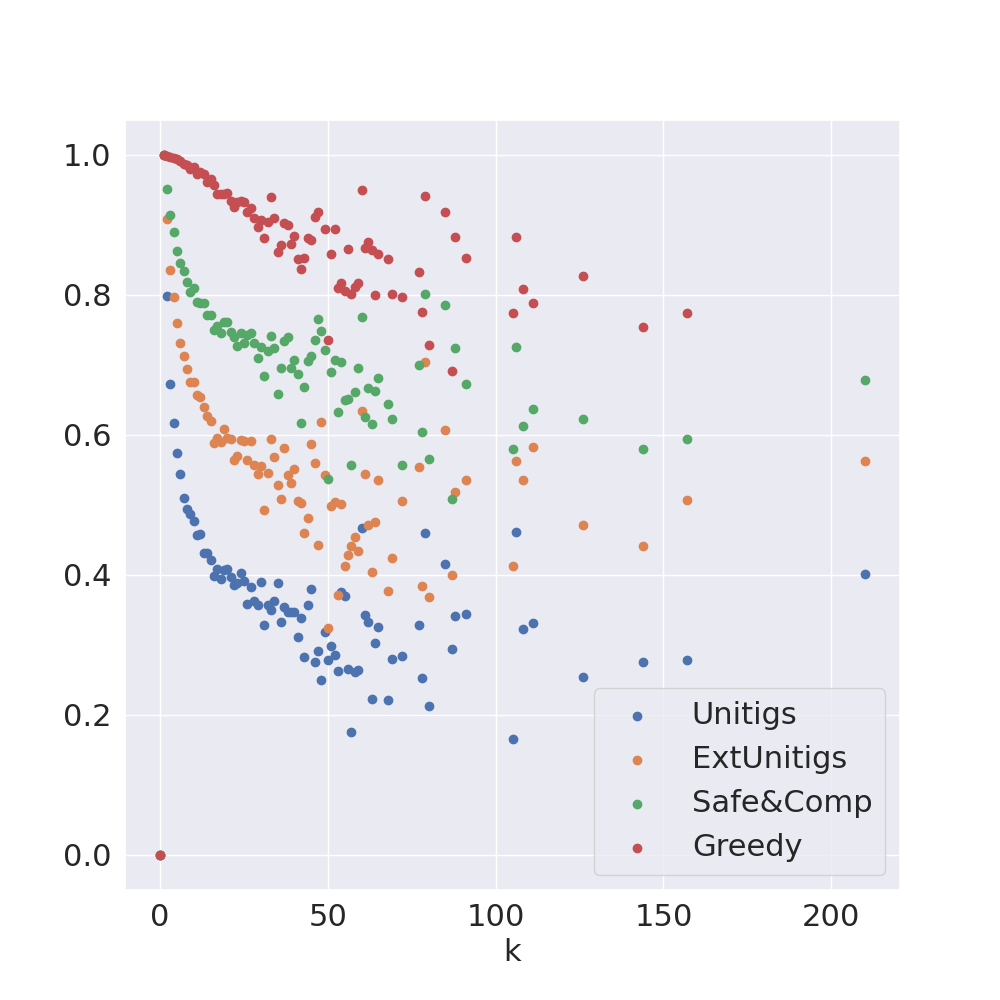}
         \caption{Maximum Relative Coverage}
         \label{fig:Funnel_dataset_2_coverage}
     \end{subfigure}
    \begin{subfigure}[b]{0.32\textwidth}
         \centering
         \includegraphics[trim=30 25 75 90,clip,  width=\textwidth]{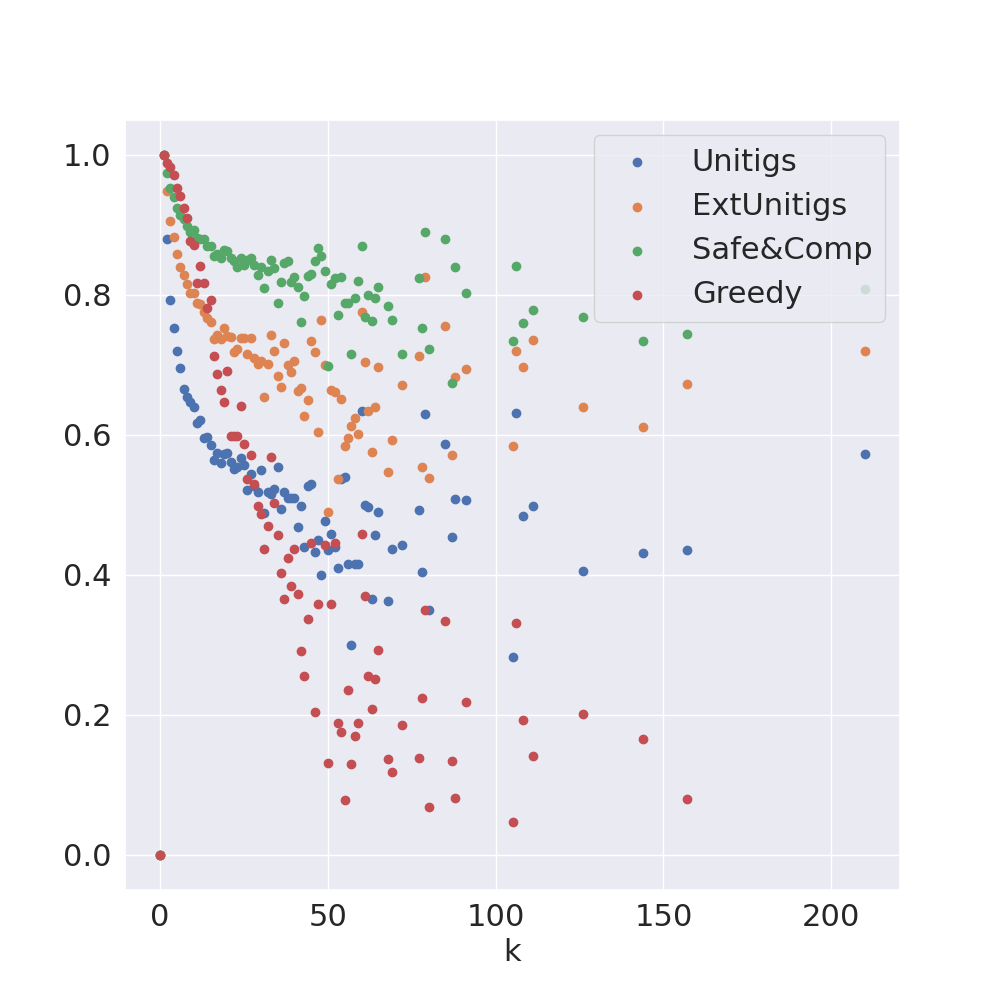}
         \caption{F-Score}
         \label{fig:Funnel_dataset_2_fscore}
     \end{subfigure}
\caption{Evaluation metrics on graphs distributed by $k$ for the {\em complete} (including funnels) Reference-Sim dataset.}
   \label{fig:Funnel_dataset_2}
\end{figure}

\begin{table}
\centering
\begin{tabular}{|c|c|c|c|c|c|c|}
\hline 
Graphs & Algorithm & Max. Coverage & Wt. Precision &  F-Score\\ \hline 
\multirow{3}{*}{\parbox[c]{2cm}{\begin{center}$k\geq 1$\\(100\%)\end{center}}} & Unitigs & 0.73 & 1.00 & 0.82\\
& ExtUnitigs & 0.84 & 1.00 & 0.90\\
& Safe\&Comp & 0.91 & 1.00 & 0.95\\
& Greedy & 0.99 & 0.91 & 0.93\\
\hline
\multirow{3}{*}{\parbox[c]{2cm}{\begin{center}$1\le k \le 10$\\(85\%)\end{center}}} & Unitigs & 0.79 & 1.00 & 0.86\\
& ExtUnitigs & 0.88 & 1.00 & 0.93\\
& Safe\&Comp & 0.93 & 1.00 & 0.96\\
& Greedy & 1.00 & 0.96 & 0.97\\
\hline
\multirow{3}{*}{\parbox[c]{2cm}{\begin{center}$k>10$\\(15\%)\end{center}}} & Unitigs & 0.41 & 1.00 & 0.58\\
& ExtUnitigs & 0.61 & 1.00 & 0.75\\
& Safe\&Comp & 0.76 & 1.00 & 0.86\\
& Greedy & 0.95 & 0.61 & 0.70\\

\hline
\end{tabular}
\caption{Summary of evaluation metrics for the {\em complete} (including funnels) Reference-Sim dataset.}
\label{tab:human_funnels}
\end{table}

\begin{figure}
     \centering
     \begin{subfigure}[b]{0.32\textwidth}
         \centering
         \includegraphics[trim=30 25 75 90,clip,  
         width=\textwidth]{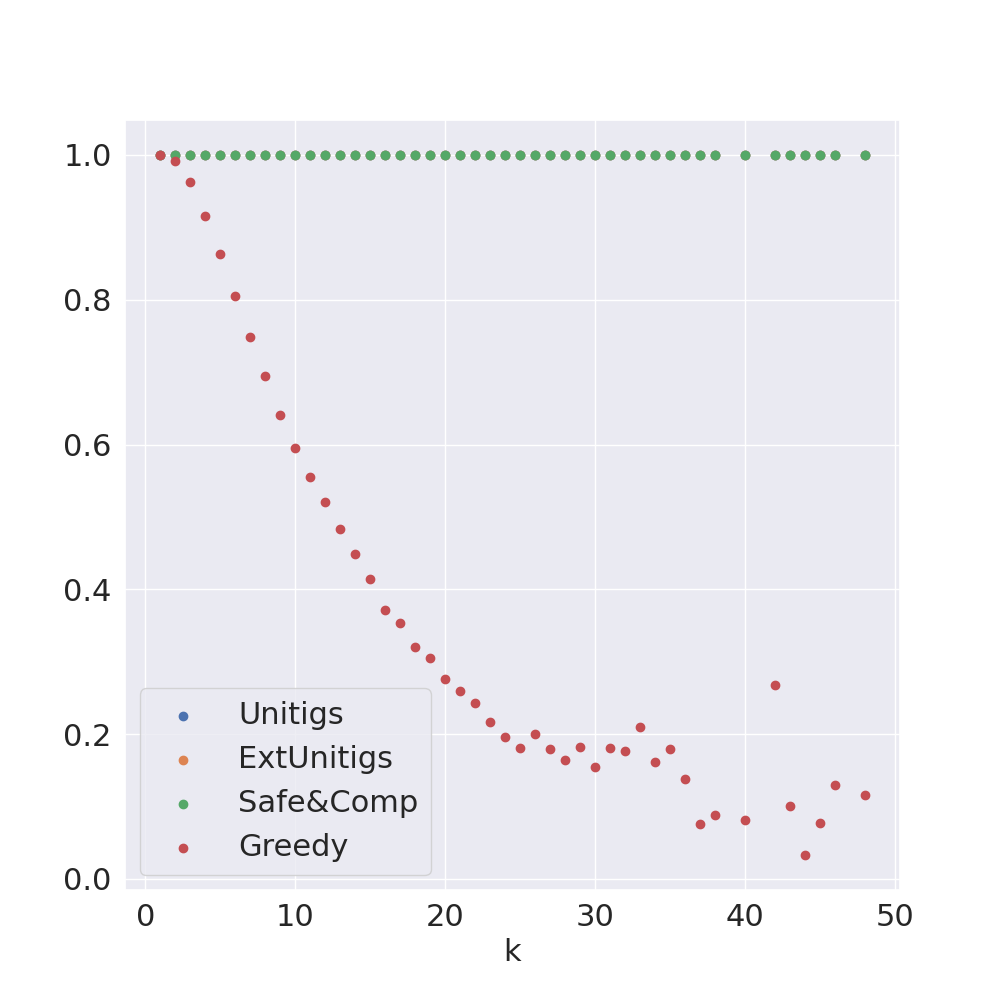}
         \caption{Weighted Precision}
          \label{fig:Funnel_dataset_1_precision}
     \end{subfigure}
     \begin{subfigure}[b]{0.32\textwidth}
         \centering
         \includegraphics[trim=30 25 75 90,clip,  width=\textwidth]{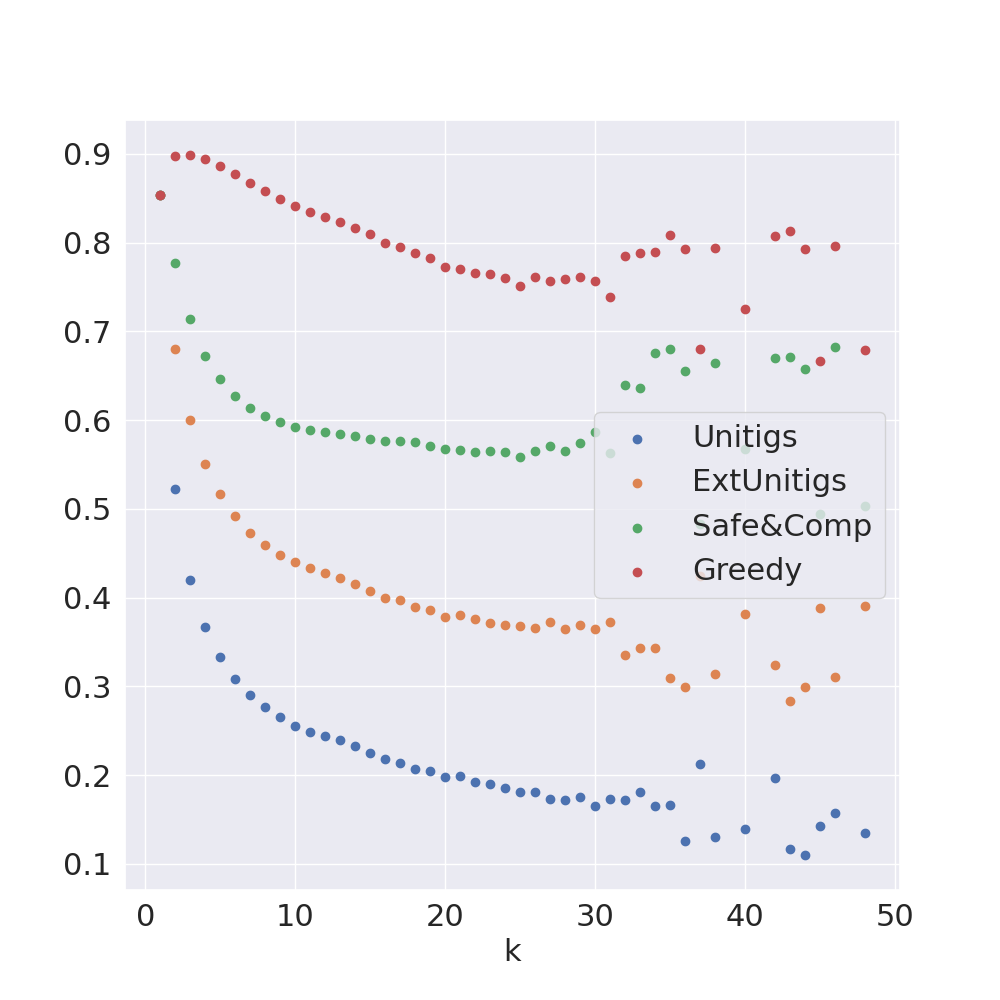}
         \caption{Maximum Relative Coverage}
         \label{fig:Funnel_dataset_1_coverage}
     \end{subfigure}
    \begin{subfigure}[b]{0.32\textwidth}
         \centering
         \includegraphics[trim=30 25 75 90,clip,  width=\textwidth]{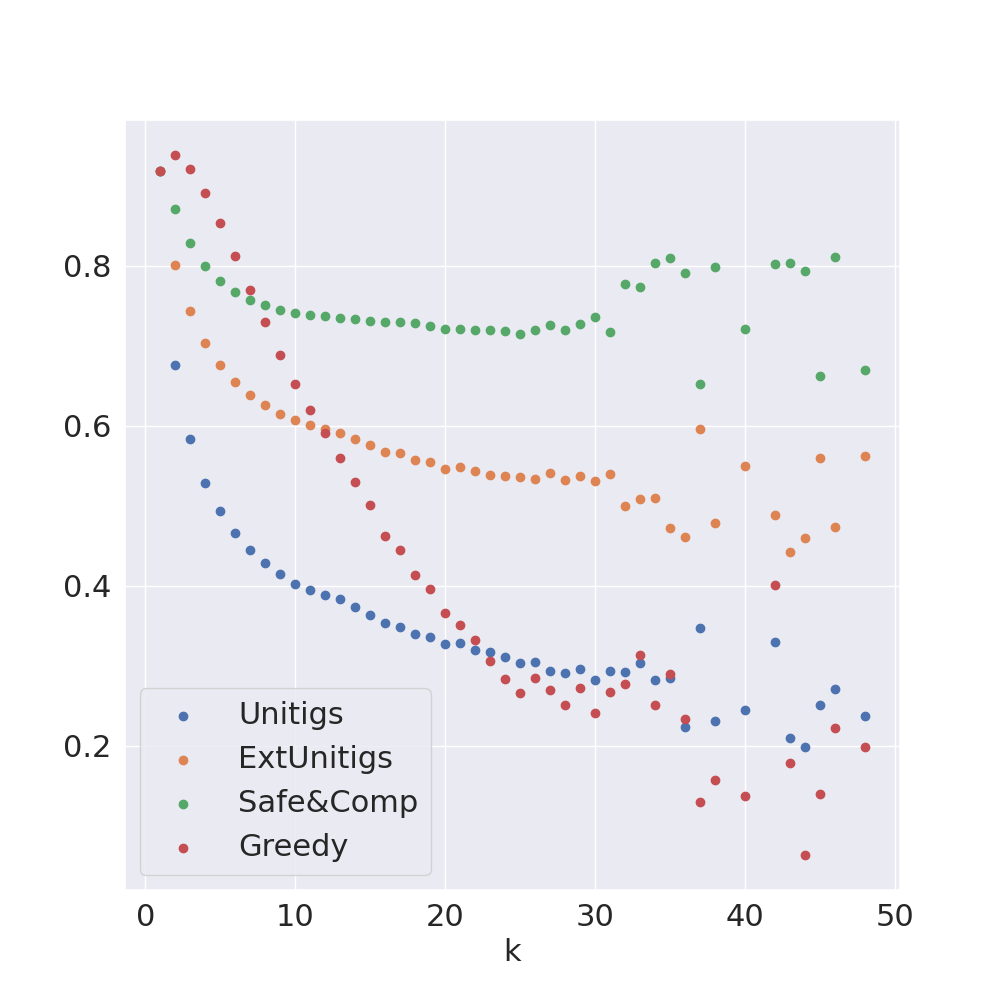}
         \caption{F-Score}
         \label{fig:Funnel_dataset_1_fscore}
     \end{subfigure}
\caption{Evaluation metrics on graphs distributed by $k$ for the {\em complete} (including funnels) Catfish dataset.}
   \label{fig:Funnel_dataset_1}
\end{figure}

\begin{table}
\centering
\begin{tabular}{|c|c|c|c|c|c|c|}
\hline 
Graphs & Algorithm & Max. Coverage & Wt. Precision &  F-Score\\ \hline 
\multirow{3}{*}{\parbox[c]{2cm}{\begin{center}$k\geq 1$\\(100\%)\end{center}}} & Unitigs & 0.64 & 1.00 & 0.75\\
& ExtUnitigs & 0.73 & 1.00 & 0.83\\
& Safe\&Comp & 0.78 & 1.00 & 0.87\\
& Greedy & 0.87 & 0.96 & 0.91\\
\hline
\multirow{3}{*}{\parbox[c]{2cm}{\begin{center}$1\le k \le 10$\\(99\%)\end{center}}} & Unitigs & 0.65 & 1.00 & 0.76\\
& ExtUnitigs & 0.73 & 1.00 & 0.83\\
& Safe\&Comp & 0.79 & 1.00 & 0.88\\
& Greedy & 0.87 & 0.97 & 0.91\\

\hline
\multirow{3}{*}{\parbox[c]{2cm}{\begin{center}$k>10$\\(1\%)\end{center}}} & Unitigs & 0.24 & 1.00 & 0.38\\
& ExtUnitigs & 0.42 & 1.00 & 0.59\\
& Safe\&Comp & 0.58 & 1.00 & 0.74\\
& Greedy & 0.82 & 0.48 & 0.56\\

\hline
\end{tabular}
\caption{Summary of evaluation metrics for the {\em complete} (including funnels) Catfish dataset, computed relative to nodes.}
\label{tab:catfish_funnels}
\end{table}

\end{document}